\documentclass{amsart}
 \usepackage[utf8]{inputenc}
\usepackage{latexsym,amsfonts}
\usepackage{tikz}
\usetikzlibrary{calc,shapes,arrows,automata}
\usepackage{amsmath}
\usepackage{algorithm}
\usepackage{algorithmicx}
\usepackage{url}
\usepackage{subcaption} 
\usepackage{caption}
\usepackage{algpseudocode}
\newcommand{\R}{\mathbb{R}}
\newcommand{\N}{\mathbb{N}}

\newcommand{\Bb}{\mathcal{B}}

\newcommand{\Xx}{\mathcal{X}}

\newcommand{\Sys}{\mathfrak{S}}
\newcommand{\Ll}{\mathcal{L}}

\newcommand{\Oo}{\mathcal{O}}
\newcommand{\Inf}{\mathsf{Inf}}
\newcommand{\Traces}{\mathsf{Traces}}

\newcommand{\Aa}{\mathcal{A}}
\newcommand{\hX}{\hat{\Xx}}
\newcommand{\hf}{\hat{f}}
\newcommand{\hSys}{\hat{\Sys}}
\newcommand{\hx}{\hat{x}}
\newcommand{\Orac}{\Oo_{\Xx_{VF}}}
\newcommand{\VM}[1]{{#1}}

\newtheorem{theorem}{Theorem}[section]

\newtheorem{problem}[theorem]{Problem}
\newtheorem{lemma}[theorem]{Lemma}
\newtheorem{corollary}[theorem]{Corollary}

\newtheorem{definition}[theorem]{Definition}

\title{Co-B\"uchi Barrier Certificates for Discrete-time Dynamical Systems}

\begin{document}

\thanks{
Corresponding author Vishnu Murali.
This work was supported by the NSF grants ECCS-2015403, CNS-2039062, and CNS-2111688.
}
\author[vishnu murali]{Vishnu Murali$^1$}
\author[ashutosh trivedi]{Ashutosh Trivedi$^1$}
\author[majid zamani]{Majid Zamani$^1$}

\address{$^1$University of Colorado Boulder}
\email{\{vishnu.murali,ashutosh.trivedi,majid.zamani\}@colorado.edu }
\urladdr{}
\urladdr{}
\urladdr{}
\begin{abstract}
Barrier certificates provide functional overapproximations for the reachable set of dynamical systems and provide inductive guarantees on the safe evolution of the system.
In  automata-theoretic verification, a key query is to determine whether the system visits a given predicate over the states ``finitely often'', resulting from the complement of the traditional B\"uchi acceptance condition. 
This paper proposes a barrier certificate approach to answer such queries by developing a notion of \emph{co-B\"uchi barrier certificates} (CBBCs) that 
generalize classic barrier certificates to ensure that the traces of a system visit a given predicate a fixed number of times.
Our notion of CBBC is inspired from \emph{bounded synthesis} paradigm to LTL realizability,  where the LTL specifications are converted to safety automata  via (universal) co-B\"uchi automata with a bound on final state visitations provided as a hyperparameter.
Our application of CBBCs in verification is analogous: we fix a bound and search for a suitable barrier certificate, increasing the bound if no suitable function can be found.
We then use these CBBCs to verify our system against properties specified by co-B\"uchi automata and demonstrate their effectiveness via some case studies. 
\end{abstract}

\maketitle

\section{Introduction}
\label{sec:intro}
The notion of \emph{barrier certificates}~\cite{prajna_2002_introducing}
combined with constraint satisfaction based search approaches have transformed automatic safety verification for dynamical systems.
A barrier certificate is a function over the state set of the system, such that the value of the barrier is non-positive over the initial set of states, is non-increasing as the system evolves, and positive over the unsafe set of states.
They thus act as a functional overapproximation of the reachable set of states where any state is in the overapproximation if it has a value that is non-positive.
When combined with a positivity condition over the unsafe set, it guarantees safety.
In the classic automata-theoretic approach to verification, the central query involves ascertaining whether a set of states may only be visited finitely often. 
In particular, bounded synthesis based approaches ~\cite{schewe_2007_bounded,filiot_2009_antichain,filiot_2011_antichains,bohy_2012_acacia+} are centered around the idea that the problem of determining whether a set is visited only finitely often can be answered by demonstrating a safety condition; this condition characterizes that the number of times the system can reach the aforementioned set is bounded from above by some fixed hyperparameter.
\emph{This paper presents an approach to automata-theoretic verification of discrete-time dynamical systems by advocating a notion of Co-B\"uchi barrier certificates. }

\noindent \textbf{Contributions.}
The key idea of the paper builds around the following observation: while the question of visiting a region only finitely often is not a direct safety property for the system, if one instead considers a related system that appends the original system with a counter keeping track of the region visitation, a bounded-visitation property can be cast as a safety property over the extended system. 
Moreover, if the extended system is safe with respect to the modified safety property, it provides a certificate for bounded-visitation of the region for the original system.
Of course, we need not explicitly compute this modified system; instead we may simply modify the definition of barrier certificates over the original system.
We dub this modified definition \emph{Co-B\"uchi barrier certificates} (CBBCs).
We build on this idea, and propose an automata-theoretic technique to verify discrete-time dynamical systems against properties specified as universal co-B\"uchi automata (UCA).
We present two computational methods (SMT and Sum-of-Squares based) to  search for these CBBCs.

\noindent \textbf{Related works.}
Traditional model checking approaches~\cite{henzinger_1997_hytech,tabuada_2009_verification,baier_2008_principles} for continuous space systems against temporal properties, such as those specified in linear temporal logic~\cite{pnueli_1977_temporal}, typically involve building and verifying their finite-state abstractions against these properties.
However, these approaches suffer from the curse of dimensionality: their complexity is exponential in the dimension of the state space.
First proposed by Prajna et al.~\cite{prajna_2004_safety}, the notion of barrier certificates has seen widespread use~\cite{prajna_2007_convex,ames_2016_control,clark_2021_control} in the verification of safety for continuous-space systems as they provide an abstraction-free approach. 
While these approaches are sound, techniques to find barrier certificates of a certain template (\textit{e.g.} a polynomial) that rely on semidefinite programming~\cite{boyd_2004_convex} such as sum-of-squares~\cite{Parrilo_2003}, have a complexity that is polynomial in the dimension of the state space (assuming their template is fixed).

Wongpiromsarn et al.~\cite{wongpiromsarn_2015_automata} introduced the idea of using barrier certificate to verify continuous-space systems against temporal logic properties and subsequently it has inspired several extensions~\cite{dimitrova_2014_deductive,papusha_2016_automata,jagtap_2018_temporal,jagtap_2019_formal,anand_2022_small}.
A key feature of the  approaches from~\cite{wongpiromsarn_2015_automata,jagtap_2019_formal} is the idea of ``cutting'' the automaton (by breaking it into state triplets) using barrier certificates; we call this approach the state triplet approach in our work.
As an example consider a system, a labeling function mapping states to observations, and a specification.
The complement of the specification is shown in Figure~\ref{fig:aut_eg_1} and our goal is to find a barrier certificate that cuts every path from the initial state to the accepting state by demonstrating that some sequences of transitions within this path is impossible.
For instance, in Figure~\ref{fig:aut_eg_1}, we are required to ``cut'' the state triplet $(q_0, q_1, q_2)$ to disallow all traces that start from any state with the label $b$ from reaching a state with the label $a$.
\begin{figure}[t]
    \centering
    \begin{tikzpicture}[node distance =2cm]
    \node[initial, state, draw, initial text =,fill=blue!10!white] (0) at (0,0) {$q_0$};\
    \node[,state, fill=blue!10!white,] (1) at (2,0) {$q_1$};
    \node[ accepting,state, fill=blue!10!white,] (2) at (4,0) {$q_2$};
    \path[->]
    (0) edge node[above]{$b$} (1)
    (1) edge[loop below] node{$b$} (1)
    (1) edge node[above]{$a$} (2)
    (2) edge[loop right] node{$a , b$} (2);
    \end{tikzpicture}
    \caption{Example NBA $\Aa_{b_1}$ from Section~\ref{sec:intro} which represents the complement of a safety specification.}
    \label{fig:aut_eg_1}
\end{figure}
While the state triplet approach is useful in verifying properties specified by automaton on finite words (especially properties that are specified by co-safe languages~\cite{kupferman_2001_model}) , it unfortunately  faces limitations when dealing with general automata on infinite words.
Note that the acceptance condition of automata over infinite words ask whether an accepting state is visited infinitely often, however this state triplet approach may be used to verify only those systems whose traces reach the accepting state of the automaton at most twice (cf. Subsection~\ref{subsec:triplet}).
Further this approach is independent of the runs in the automaton and of the initial states of the system, in the following way:
one is required to break the edge pairs of a path regardless of what states of the automaton may be encountered before or after.
We show (cf. Subsection~\ref{sec:subsumption}) that our CBBC-based approach generalizes the state triplet approach for discrete-time dynamical systems.
The authors of~\cite{murali_2023_closure} present an approach to verify $\omega$-regular properties via transition invariants rather than state invariants. However, the complexity of the certificates is dependent on twice the number of dimensions of the system.
Our notion of CBBCs is inspired by the ideas present in bounded synthesis~\cite{schewe_2007_bounded,filiot_2009_antichain,filiot_2011_antichains,bohy_2012_acacia+}, where one successively increases the bound on the number of times a region is visited.
While the approach in~\cite{papusha_2016_automata} relies on the product of the system with the automaton and is thus less conservative in comparison to the state triplet approach, it does not allow the same degree of giving incremental guarantees, as one needs to unfold the automaton to increase the bound on the number of visits.
The authors of \cite{dimitrova_2014_deductive} provide a notion of parity certificates to provide guarantees against parity automata via a combination of inductive invariants and ranking functions.
CBBCs are not only sufficient but also necessary and rely only on safety arguments.
Our notion of CBBCs have been extended to provide a compositional approach to synthesize controllers for large-scale  systems represented as interconnections of smaller subsystems in \cite{galarza-jimenez_2024} as well as for data-driven and stochastic systems \cite{ajeleye2024data,ajeleye2024co}.

\section{Preliminaries}
\label{sec:prelims}
We use $\N$ and $\R$ to denote the set of natural numbers and reals.
For $a \in \R$, we use $\R_{\geq a}$ and $\R_{> a}$ to denote the intervals $[a, \infty)$ and $(a,\infty)$ respectively, and similarly, for any natural number $n \in \N$, we use $\N_{\geq n}$ to denote the set of natural numbers greater than or equal to $n$.
\VM{We use $[0..k]$ to denote the set $\{0,\ldots, k \}$ for any $k \in \N$.}
Given a set $A$, sets $A^{*}$ and $A^{\omega}$ denote the set of finite and countably infinite sequences of elements in $A$, while $|A|$ denotes the cardinality of the set. 
Given two sets $A,B$, we use $A \setminus B$ to denote the set containing those elements that are present in $A$ but not in $B$, and as usual use $A \cup B$ and $A \cap B$ to represent their union and intersection.
If $A \subseteq B$, and the set $B$ can be inferred from the context, we denote the complement $B \setminus A$ simply as $\overline{A}$. 

Given a function $f:A \to A$, we define $f^i:A \to A$ recursively as $f^0(a) = a$ for any $a \in A$, and $f^{i+1}(a) = f^i(f(a))$ for all $i \in \N$.
Given a relation $R \subseteq A \times B$, and an element $a \in A$, let $R(a)$ denote the set $\{ b \mid (a,b) \in R \}$. 
We use $(a_1, a_2, \ldots, a_n) \in A^{*}$ to indicate a finite-length sequence of length $n \in \N$, and $\langle a_0, a_1, \ldots  \rangle\in A^{\omega}$ to indicate a countably infinite sequence of elements from the set $A$.
We say that an infinite sequence $s = \langle a_0, a_1, \ldots \rangle\in A^{\omega}$, visits an element $a \in A$, at most $n$ times if there exists at most $n$ distinct indices $i_0, i_1, \ldots, i_{(n - 1)} \in \N$ with $a_{i_j}  = a $ for all $j \in [0..(n-1)]$, and $a_\ell \neq a$ for all $\ell \in \N \setminus \{i_0, i_1 ,\ldots, i_{(n - 1)} \}$.
This generalizes to determining whether a sequence $s$ visits a subset $B \subseteq A$ at most $n$ times.
Let $\Inf(s)$ be the set of elements that occur infinitely often in the sequence $s$.

\subsection{Automata over Infinite Words}
\label{subsec:prelims_aut}
An automaton $\Aa$ is a tuple $(\Sigma,Q, Q_0, \delta, Q_{Acc})$, where $\Sigma$ denotes a finite alphabet, $Q$ a finite set of states, $Q_0 \subseteq Q$ an initial set of states, $\delta \subseteq Q \times \Sigma \times Q$ the transition relation, and $Q_{Acc} \subseteq Q$ denotes a set of accepting states.
A run of the automaton $\Aa = (\Sigma,Q, q_0, \delta, Q_{Acc})$ over a word $w = \langle \sigma_0, \sigma_1, \sigma_2, \ldots \rangle \in \Sigma^{\omega}$, is an infinite sequence of states $\rho = \langle q_0,q_1, q_2, \ldots \rangle \in Q^{\omega}$ with $q_0 \in Q_0$ and $q_{i+1} \in \delta(q_i, \sigma_i)$.
A \emph{nondeterministic B\"uchi automaton} (NBA) $\Aa_b = (\Sigma,Q, Q_0, \delta, Q_{Acc})$ is said to accept a word $w$, if there exists a run $\rho$ on $w$ where $\Inf(\rho) \cap Q_{Acc} \neq \emptyset$.
A \emph{Universal Co-B\"uchi automaton} (UCA) $\Aa_c = (\Sigma,Q, Q_0, \delta, Q_{Acc})$ accepts a word $w$, if for every run $\rho$ on $w$ we have $\Inf(\rho) \cap Q_{Acc} = \emptyset$.
Our definition of a UCA is based on the definition considered in~\cite{filiot_2011_antichains}, though alternative, and equivalent semantics exist~\cite{esparza_2014_ltl}.
For $k \in \N$, a $k$-Universal  Co-B\"uchi automaton ($k$-UCA) $\Aa_{c,k} = (\Sigma,Q, Q_0, \delta, Q_{Acc})$ accepts a word $w$ if, every run $\rho$ on $w$ visits some state in the set $Q_{Acc}$ at most $k$ times.
The language of an automaton $\Aa$ denoted $L(\Aa)$ is defined as the set of words accepted by the automaton, \textit{i.e.}, $L(\Aa) = \{ w \mid \Aa \text{ accepts word } w \}$.
It is well known that NBAs are closed under complementation~\cite{safra_1988_complexity}: given an NBA $\Aa_b = (\Sigma,Q, Q_0, \delta, Q_{Acc})$, there exists an NBA $\Aa'_b = (\Sigma, Q', Q'_0, \delta', Q'_{Acc})$ such that $L(\Aa_b) = \overline{L(\Aa'_b)}$.
We also observe that the acceptance conditions of an NBA and a UCA are duals of each other, \textit{i.e.}, given an NBA $A_b =  (\Sigma,Q, Q_0, \delta, Q_{Acc})$, and a UCA $\Aa_c =  (\Sigma,Q, Q_0, \delta, Q_{Acc})$, with the same structure (the same sets and transition relation), we have $L(\Aa_b) = \overline{L(\Aa_c)}$. 
Observe that for natural numbers $k_1, k_2 \in \N$, a $k_1$-UCA $\Aa_{c, k_1} =  (\Sigma,Q, Q_0, \delta, Q_{Acc})$, and a $k_2$-UCA $\Aa_{c, k_2} =  (\Sigma,Q, Q_0, \delta, Q_{Acc})$ with the same structure, we have $L(\Aa_{c, k_1}) \subseteq L(\Aa_{c, k_2})$.
Finally, for a $k$-UCA $\Aa_{c, k} =  (\Sigma,Q, Q_0, \delta, Q_{Acc})$ and the UCA $\Aa_{c} =  (\Sigma,Q, Q_0, \delta, Q_{Acc})$ we have $L(\Aa_{c, k_1}) \subseteq L(\Aa_{c})$.

\subsection{Discrete-time Dynamical System}
\label{subsec:prelims_system}
A discrete-time dynamical system (or simply, a system) $\Sys$ is a tuple $(\Xx,\Xx_0, f)$, where $\Xx \subseteq \R^n$ denotes the state set, $\Xx_0 \subseteq X$ denotes a set of initial states, and $f:\Xx \to \Xx$ the state transition function.
The state evolution of the system is given by $\Sys: x(t+1) = f(x(t)).$

A \emph{state sequence} is an infinite sequence  $\langle x_0, x_1, \ldots \rangle \in \Xx^{\omega}$ where $x_0 \in \Xx_0$, and $x_{i+1} = f(x_i)$, for all $i \in \N$.
We associate a labeling function $\Ll: \Xx \to \Sigma$ which maps each state with a letter in a finite alphabet. 
This generalizes to mapping a state sequence of the system $\langle x_0, x_1, \ldots \rangle \in \Xx^{\omega}$ to a a trace or word $w = \langle \Ll(x_0), \Ll(x_1), \ldots \rangle \in \Sigma^{\omega}$. 
Finally, let $\Traces(\Sys, \Ll)$ denote the set of all traces of system $\Sys$ under the labeling map $\Ll$. 

\subsection{Safety Verification and Barrier Certificates}
\label{subsec:prelims_safety}
A system $\Sys$ is safe with respect to a set of unsafe states $\Xx_u$, if no state sequence reaches $\Xx_u$, \textit{i.e.}, for every state sequence  $\langle x_0, x_1, \ldots  \rangle$, we have $x_i \notin \Xx_u$ for all $i \in \N$.
\begin{definition}
\label{def:bar}
    A function $\Bb: \Xx \to \R$ is a barrier certificate for a system $\Sys$ with respect to a set of unsafe states $\Xx_u$ if: 
    \begin{align}
        & \Bb(x) \leq 0, \text{ for all } x \in \Xx_0, \label{eq:bar_cond_1} \\
        & \Bb(x) > 0, \text{ for all } x \in \Xx_u, \text{ and} \label{eq:bar_cond_2} \\
        & \Bb(f(x)) - \Bb(x) \leq 0, \text{ for all } x \in \Xx. \label{eq:bar_cond_3}
    \end{align}
\end{definition}
The existence of a barrier certificate ensures the safety of a system as discussed below.

\begin{theorem}
\label{thm:bar_safe}
    Consider a system $\Sys {=} (\Xx, \Xx_0, f)$, with a set of unsafe states $\Xx_u$. The existence of a function $\Bb:\Xx \to \R$ as in Definition \ref{def:bar} implies that the system is safe.
\end{theorem}
The proof of Theorem \ref{thm:bar_safe} follows from \cite{prajna_2004_safety}.
We should note that while barrier certificates are a sufficient proof of safety for continuous-time systems, they are not only sufficient but also necessary for discrete-time systems.
\begin{lemma}
    Consider a system $\Sys = (\Xx, \Xx_0, f)$ that is safe with respect to a set of unsafe states $\Xx_u$.
    Then there exists a barrier certificate $\Bb: \Xx \to \R$ as in Definition \ref{def:bar}.
\end{lemma}
\begin{proof}
Consider an oracle function $\Oo: 2^\Xx \to 2^\Xx$ defined as  $\Oo(S) = \{x' \mid \text{ there exists } i \in \N, \text{ such that } x' = f^i(x) \text{ for some } x \in S  \}$ for all sets $S \subseteq \Xx$.
    Observe that the set $\Oo( \Xx_0)$ denotes the set of all reachable states of the system.
    As the system is safe, we must have $\Oo( \Xx_0) \cap \Xx_u = \emptyset$.
    Consider the function $\Bb: \Xx \to \R$ defined as:
    $\Bb(x) =\begin{cases}
      0 & \text{ if } x \in \Oo(\Xx_0) \\
      1 & \text{ } \text{ otherwise }
    \end{cases}$.
    We now show how this function is a barrier certificate as in Definition \ref{def:bar}.
    Observe that conditions \eqref{eq:bar_cond_1} and \eqref{eq:bar_cond_2} trivially hold as $\Xx_0 \in \Oo(\Xx_0)$ by definition and $\Xx_u \cap \Oo(\Xx_0) = \emptyset$ from the assumption that the system is safe.
    Now, let us consider \eqref{eq:bar_cond_3}.
    The only way for \eqref{eq:bar_cond_3} to not hold, is if we have $\Bb(f(x)) > \Bb(x)$, in such a case we must have $\Bb(x) = 0$, and $\Bb(f(x)) = 1$. 
    If $\Bb(f(x)) = 1$, then $f(x) \notin \Oo( \Xx_0 )$, while $\Bb(x) = 0$ implies that $x \in \Oo(\Xx_0)$.
    If the state $x \in \Oo(\Xx_0)$, by the definition of oracle $\Oo$, there must exist $x_0 \in \Xx_0$, and $i \in \N$, such that $x = f^i(x_0)$.
    Furthermore $f(x) = f^{i+1}(x)$, and so we must have $f(x) \in \Oo(x)$.
    This is a contradiction. \qed
\end{proof}
We now consider the problem of verifying more complex $\omega$-regular objectives via barrier certificates.

\subsection{Problem Definition}
\label{subsec:problem}

We now focus on the problem of characterizing barrier certificates to demonstrate that a set of states is visited only finitely often, and use it to provide a barrier certificate based approach for automata-theoretic verification.
Particularly, we wish to define effective notions of barrier certificates to demonstrate that every state sequence of $\Sys = (\Xx, \Xx_0, f)$ visits a region $\Xx_{VF} \subseteq \Xx $ only finitely often.
We then use these certificates to verify systems against properties specified by automata.

Given an NBA $\Aa_b =(\Sigma,Q, Q_0, \delta, Q_{Acc}) $, and UCA $\Aa_c = (\Sigma,Q, Q_0, \delta, Q_{Acc})$ we have $L(\Aa_c ) = \overline{L(\Aa_b)}$.
Thus showing that every trace of the system $\Sys$ visits $Q_{Acc}$ only finitely often is sufficient to verify that the traces of the system are in $L(\Aa_c)$ and hence not in $L(\Aa_b)$.

\begin{problem}
\label{prob:def}
    Prove that the traces $TR(\Sys, \Ll)$ of $\Sys = (\Xx, \Xx_0, f) $ under labeling $\Ll: \Xx \to \Sigma$, are in the language of a $k$-UCA $\Aa_{c, k} = (\Sigma, Q, Q_0, \delta, Q_{Acc})$.
\end{problem}

\begin{corollary}
\label{cor:kUCA_UCA}
    If  the traces $TR(\Sys, \Ll)$ of a system $\Sys = (\Xx, \Xx_0, f) $ under a labeling map $\Ll: \Xx \to \Sigma$, are contained in the language of a $k$-UCA $\Aa_{c, k} = (\Sigma, Q, Q_0, \delta, Q_{Acc}) $, then they are contained in the language of the UCA $\Aa_c = (\Sigma, Q, Q_0, \delta, Q_{Acc}) $.
\end{corollary}

\section{Co-B\"uchi Barrier Certificates (CBBCs)} 
\label{sec:finite_visits}
Barrier certificates act as proofs for safety, however many properties of interest require showing that a given region is visited a finite number of times rather than never visited at all.
In particular, the accepting conditions of NBAs and UCAs rely on showing that the accepting states are visited infinitely or only finitely often.
We thus consider approaches where one may use barrier certificates to show a set $\Xx_{VF}$ is visited only finitely often.

\VM{Our notion of CBBCs are inspired by the bounded synthesis paradigm \cite{filiot_2011_antichains}, where one fixes the number of times an accepting state is visited and tries to prove that the system does not visit an accepting state beyond this bound.
Analogously, we fix a counter to keep track of the number of times a state in $\Xx_{VF}$ is seen. We then modify the system $\Sys = (\Xx, \Xx_0, f)$ by appending this counter variable to the state set.
Formally, we construct a system $\Sys' = (\Xx', \Xx'_0, f')$, where $\Xx' = \Xx \times \N$ denotes the state set (containing an additional counter value).
The set $\Xx_0' := \Xx_0 \times \{ 0 \} $
consist of the initial set of states, and the transition function $f': \Xx' \to \Xx'$ is defined as:
\begin{equation}
\label{eq:counter_trans}
    f'((x,\ell)) = 
\begin{cases}
(f(x), \ell+1) & \text{ for all } x \in \Xx_{VF}, \\
(f(x), \ell) & \text{ otherwise},
\end{cases}
\end{equation}
}
\VM{for all $i \in \N$.}
Observe that any state sequence $\langle x_0, x_1, \ldots \rangle$ of the system $\Sys$ that visits the region $\Xx_{VF}$ only finitely often must have some index $i \in N$, with $x_j \notin \Xx_{VF}$ for all $j \geq i$.
This corresponds to a state sequence $\langle x'_0, x'_1, \ldots \rangle$ of the system $\Sys'$ such that there exists some $k \in \N$ with $x'_{j+1} = (f(j),k )$ for all $j \geq i$.
Therefore showing the system $\Sys$ visits $\Xx_{VF}$ only finitely often reduces to showing the system $\Sys'$ never reaches a state $(x, k+1)$ for some $k \in \N$, and some $x \in \Xx_{VF}$.
We now define CBBCs over the original system by taking this modified system's dynamics into account.
\begin{definition}
\label{def:cbar}
    A function $\Bb: \Xx \times [0..k] \to \R$ is a co-B\"uchi barrier certificate (CBBC) for a system $\Sys$ with set of states $\Xx_{VF}$ that must be visited only finitely often if there exists some $k \in \N$ such that: 
    \begin{align}
        & \Bb(x,0) \leq 0,  \text{ for all } x \in \Xx_0, \label{eq:cbar_cond_1} \\
        & \VM{\Bb(x, k) > 0,\text{ for all } x \in \Xx_{VF}}, \label{eq:cbar_cond_3} 
    \end{align}
and for all values $i \in [0..k]$, and \VM{$j \in [0..(k-1)]$} we have:
\begin{align}
            & \Bb(f(x),i) - \Bb(x,i) \leq 0,  \VM{\text{ for all } x \notin \Xx_{VF}} \text{ and, } \label{eq:cbar_cond_4} \\
        & \Bb(f(x),j+1) - \Bb(x,j) \leq 0, \VM{\text{ for all  } x \in \Xx_{VF}}. \label{eq:cbar_cond_5} 
\end{align}
\end{definition}
\begin{theorem}
\label{thm_fin_visit}
    Consider a system $\Sys = (\Xx, \Xx_0, f)$, with a region $\Xx_{VF}$. Suppose there exists a $k \in N$, and a function $\Bb:\Xx \times [0..k] \to \R$ as in Definition \ref{def:cbar}. Then the traces of the system visits the region $\Xx_{VF}$ at most $k$ times.
\end{theorem}
\begin{proof}
\VM{
    We prove this by contradiction.
    Suppose there exists a state sequence $\langle x_0, x_1, \ldots \rangle $ of system $\Sys$ that visits the region $\Xx_{VF}$ more than $k$ times. 
    In particular, let $t \in \N$ be the index at which the state sequence visits $\Xx_{VF}$ for the $(k+1)$-th time.
    This corresponds to a state sequence $\langle x'_0, x'_1, \ldots \rangle$ for system $\Sys'$ such that $x_t = (x_t, k)$, and $x'_{t+1} = (x_{t+1}, k+1)$.
    Following condition \eqref{eq:cbar_cond_1}, observe that we have $\Bb(x_0,0) \leq 0$.
    For each state $x'_i = (x_i, \ell_i)$, either $x_i \in \Xx_{VF}$, in which case, we have $x'_{i+1} =(x_{i+1}, \ell_i + 1) $, and condition \eqref{eq:cbar_cond_5} holds, or $x'_{i+1} = (x_{i+1}, \ell_i)$, and condition \eqref{eq:cbar_cond_4} holds via induction for any $i \in \N$.
    Thus, we have $\Bb(x_i,\ell_i) \leq 0$ for all $i \in \N$.
    Therefore $\Bb(x_{t}, k) \leq 0$, but this contradicts with condition~\eqref{eq:cbar_cond_3} which requires $\Bb(x_i, k+1) > 0$.
    \qed
}
\end{proof}
Observe that if we set $k = 0$ and consider $\Xx_u = \Xx_{VF}$, we recover the original formulation of barrier certificates as in Definition~\ref{def:bar}.  Here conditions~\eqref{eq:cbar_cond_1} corresponds to condition~\eqref{eq:bar_cond_1} over the initial set of states, condition~\eqref{eq:cbar_cond_3} corresponds to condition~\eqref{eq:bar_cond_2} over $\Xx_u$, and
conditions~\eqref{eq:cbar_cond_4} corresponds to condition~\eqref{eq:bar_cond_3}.
We now show how CBBCs are not only sufficient but also necessary with respect to oracles.
To do so we define a notion of oracle functions that keep track of whether a state is visited or not, and another oracle function to determine the number of times a state visits $\Xx_{VF}$.
\VM{For easier presentation, we first define these functions as follows. Let oracle function $\Oo: \Xx \to 2^\Xx$ be defined as  $\Oo(x) = \{f^i(x) \mid i \in \N \}$ for all states $x \subseteq \Xx$.
To keep track of number of visits, we define oracle function $\Orac: \Xx \to {\N \cup \{\infty \}}$ for any $x \in \Xx$ as:
\begin{itemize}
  \item $\Orac(x) :=  \infty$ , if, for all $i \in \N$, there exists $j \in \N_{\geq i}$ such that $f^j(x) \in \Xx_{VF}$.
  \item  $\Orac(x) := j$ if there are at most $j$-indices $i_1, \ldots, i_j$ such that $f^{i_k}(x) \in \Xx_{VF}$, and for all indices $\kappa \notin \N \setminus \{i_1, \ldots, i_j \}$, we have $f^\kappa(x) \notin \Xx_{VF}$. Note that $i_1 = 0$ is a valid index. That is, if $x \in \Xx_{VF}$, we have $ \Orac(x_0) \in \N_{\geq 1} \cup \{ \infty\}$ as $f^0(x) = x$.
  \item $\Orac(x) := 0$ if for all $i \in \N$, we have $f^i(x) \notin \Xx_{VF}$.
\end{itemize}
}
\VM{
Intuitively, the above oracle functions  $\Oo$ and $\Orac$ determine if some state is reachable, as well as the maximum number of times that the state $x$ visits the set $\Xx_{VF}$.
observe that the set $Reach := \underset{x_0 \in \Xx_0}{\bigcup}\Oo(x_0)$ is the set of states that are reachable from the initial state set while $\Orac(x)$ for any $x \in Reach$ denotes the number of times $x$ visits the set $\Xx_{VF}$.
}
\VM{
\begin{lemma}
\label{lem:rel_comp_cbbc}
    Consider a system $\Sys = (\Xx, \Xx_0, f)$ with a set of states $\Xx_{VF}$ that is visited only finitely often.
    Then there exists $k \in \N$, and a CBBC $\Bb: \Xx \times [0..k] \to \R$ as in Definition \ref{def:cbar}.
\end{lemma}
\begin{proof}
First, we bound the maximum number of visits to $\Xx_{VF}$ to determine the value of $k \in \N$ as follows:
Let $k = \max \{j \mid j := \Orac(x) \text{ and } x \in Reach \}$ denote the maximum number of times some state $x \in \Xx_{VF}$ may be reached from the initial set of states.
Observe that there cannot exist $k' \in \N_{> k}$ such that $k' = \Orac(x)$ for some $x \in Reach$ and so $k$ provides the maximum bound on the number of times a state $\Xx_{VF}$ may be reached.
We now define the CBBC $\Bb: \Xx \times [0..k] \to \R$  as:
$$\Bb(x,i) = \begin{cases}
    0 & \text{ if } x \in Reach, \text{ and } i + \Orac(x) \leq k, \\
    1 & \text{ otherwise. }  
\end{cases}$$
We now show that $\Bb$ is a CBBC.
Observe that condition \eqref{eq:cbar_cond_1} holds for any state $x_0 \in \Xx_0$, as $0 + \Orac(x_0) \leq k$ from our assumption that $k$ is the maximum number of times a state in $\Xx_{VF}$ may be reached.
For condition \eqref{eq:cbar_cond_3}, let us assume that $\Bb(x_{vf},k) \leq 0$ for some $x_{vf} \in \Xx_{VF}$.
Then, we must have $k + \Orac(x_{Vf}) \leq k$.
However, we have $\Orac(x_{vf}) \in \N_{ \geq 1} \cup \{\infty\}$, and so  we have 
$k + \Orac(x_{vf}) \leq k+1 \leq k$.
This is a contradiction, and so condition \eqref{eq:cbar_cond_3} must also hold.
For condition \eqref{eq:cbar_cond_4} to fail, it must be the case that $\Bb(x,i) = 0$ for some $x\in \Xx_{VF}$, and $\Bb(f(x), i+1) = 1$.
If $\Bb(f(x), i+1) = 1$, then either $f(x) \notin Reach$, or $i+1 + \Orac(f(x_{vf})) > k$.
However, we cannot have $x\in Reach$, and $f(x) \notin Reach$ by definition of Reach.
The only other option is $i+1 + \Orac(f(x)) > k$.
However, we know that $i + \Orac(x) \leq k$ as $\Bb(x,i) = 0$, and furthermore $\Orac(x) \in \N_{\geq 1}$, as $x = f^0(x)$.
Thus $\Orac(x) = t + 1$ for some $t \in \N$, and so we must have $i + t + 1 \leq k$.
Observe that if this is true, there are at most $t+1$ indices such that $f^{i_k}(x) \in \Xx_{VF}$ for some $i_1, i_2, \ldots i_{t+1} \in \N$, and in particular, we must have $i_1 = 0$ as $x\in \Xx_{VF}$. Furthermore, $f^{\ell}(x) \notin \Xx_{VF} $ for any $\ell \notin \{0, i_2 \ldots i_{t+1} \}$.
Thus, $f(x)$, has only $t$ indices $\{(i_2 - 1), \ldots, (i_{t+1} - 1) \}$ such that $f^{i_{k-1}}(f(x)) \in \Xx_{VF}$, and so we must have $\Orac(f(x)) = t$.
This is a contradiction as we cannot have $i + 1 + t > k$ from $\Bb(f(x),i) = 1$, and $i + t + 1 \leq k$ from $\Bb(x,i) = 0$.
A similar contradiction arises for condition \eqref{eq:cbar_cond_5}, however the values of $\Orac(x)$, and $\Orac(f(x))$ are the same if $x \notin \Xx_{VF}$.
\qed
\end{proof}
}

\section{Verifying Automata Properties}\label{sec:aut_verif}
We now extend the results in Section~\ref{sec:finite_visits} to verify properties specified by an automaton.
Let $\Sys = (\Xx, \Xx_0, f)$ be a dynamical system, $\Sigma$ be a finite alphabet and $\Ll:\Xx \to \Sigma$ be a labeling function. Consider $\Aa_{c,k} = (\Sigma,Q, Q_0, \delta, Q_{Acc})$ to be a $k$-UCA automaton representing a specification.
Then the goal of verification is to show that $TR(\Sys, \Ll) \subseteq L(\Aa_{c,k})$.
Observe that any trace that is accepted by the automaton visits an accepting state at most $k$-times.
This provides our approach for verifying dynamical systems against properties specified by automata.

We first take a synchronous product between the system $\Sys$ and automaton $\Aa_{c,k}$.
As the state set of the automaton is finite, we can equivalently represent the states of the automaton as  natural numbers in the set $ \{0, 1, \ldots, (|Q| - 1) \} $, \VM{represented succinctly as $[0..(|Q|-1)]$}.
We define the product as a tuple $\hSys = (\hX , \hX_0, \hf, \hX_{Acc}  )$, 
where $\hX = X \times [0..(|Q|-1)]$ denotes the state set, the initial set of states is defined as  $\hX_0 = \Xx_0 \times \{ i \mid i \in [0..(|Q|-1) \cap Q_0 \}$.
Observe that while the dynamical system has a unique successor state for any state $x \in \Xx$, the automata may be nondeterministic and so the state transition map is not a function but rather a relation for the product that is also nondeterministic. 
The state transition relation $\hf \subseteq \hX \times \hX$ is defined as follows for all $x \in \Xx$ and $i,j \in [0..(|Q|-1)]$, as 
$\hf((x,i)) = \{ (f(x), j) \mid  j \in \delta(i, \Ll(x)) \}$.
Given elements $x \in \Xx$, and $i \in [0..(|Q|-1)]$, the set $\hf((x,i))$ has at most $(|Q| -1)$ elements as $j \in [0..(|Q|-1)]$.
For convenience we use $\hX_{Acc} = \Xx \times \{ \ell \mid \ell \in Q_{Acc} \}$ to denote the set of accepting states of this product.

\label{subsec:kUCA}
A state sequence $\langle x_0, x_1 \ldots \rangle$ of system $\Sys$ is accepted by the automaton $\Aa_{c,k}$ if every corresponding state sequence $\langle \hx_0, \hx_1, \ldots \rangle $ has at most $k$-states in $\hX_{Acc}$, \textit{i.e.}, there exists $i_1, i_2, \ldots, i_{k} \in \N$, with $\hx_{i_j} \in \hX_{Acc}$, and for all $\ell \notin \{i_1, i_2, \ldots, i_{k} \}$ , we have $\hx_{\ell} \notin \hX_{Acc}$.
Our approach to verify whether a system satisfies the specification is to determine whether all traces of the product $\hSys$ visits states in the set $\hX_{Acc}$ at most $k$ times. 
If so then the system $\Sys$ satisfies the specification, and if not our approach is inconclusive.
To determine this we make use of a CBBC on the product system.

To determine this we make use of a CBBC on the product system defined as follows.
\begin{definition}
\label{def:cbar_prod}
    A function $\Bb: \Xx \times Q \times [0..k] \to \R$ is a co-B\"uchi barrier certificate for the synchronous product $\hSys$ of a system $\Sys$ and a $k$-UCA $\Aa_{c,k}$ if: 
    \begin{align}
        & \Bb(x,\ell, 0) \leq 0, \quad \text{ for all } x \in \Xx_0, \ell \in Q_0 \setminus Q_{Acc}, \label{eq:prod_cond_1} \\
        & \VM{\Bb(x,\ell, k) > 0, \; \text{ for all } x \in \Xx, \ell \in Q_{Acc}},  \label{eq:prod_cond_3} 
        \end{align}
    and for all states $x \in \Xx$, counter values $i \in [0..k]$, \VM{$j \in [0..(k-1)]$}, and all states $\ell \in [0..(|Q|-1)]$, and all states $\ell' \in \delta(\ell , \Ll(x))$, we have
    \begin{align}
        & \underset{\ell'}{\max} \Big\{ \Bb(f(x), \ell', i)  \Big \}\leq \Bb(x,\ell, i),  \VM{\text{ if } \ell \notin Q_{Acc}}, \label{eq:prod_cond_4} \\
        & \underset{\ell'}{\max} \Big \{ \Bb(f(x), \ell', j+1)  \Big\}\leq \Bb(x,\ell, j){,} \VM{\text{ if } \ell' \in Q_{Acc}}. \label{eq:prod_cond_5}
    \end{align}
\end{definition} 

\begin{theorem}
\label{thm:barrier_UCA}
Consider a system $\Sys = (\Xx, \Xx_0, f)$, a $k$-UCA $\Aa_{c,k} = (\Sigma,Q, Q_0, \delta, Q_{Acc})$, and a labeling map $\Ll: \Xx \to \Sigma$.
Suppose there exists $k \in \N$, and a function $\Bb: \Xx \times Q \times [0..k]$ as in Definition \ref{def:cbar_prod}, then $TR(\Sys, \Ll) \subseteq L(\Aa_{c,k})$.
\end{theorem}
\begin{proof}
We observe that the function $\Bb: \Xx \times Q \times [0..k] \to \R$ satisfying conditions~\eqref{eq:prod_cond_1} to~\eqref{eq:prod_cond_5} is a CBBC as in Definition~\ref{def:cbar} for the product $\hSys$ where the set of states that must be visited only finitely often $\hX_{VF}$ is $ \hX_{Acc}$.
Following Theorem~\ref{thm_fin_visit}, we note that every state sequence of of system $\hSys$ must visit the region $\hX_{Acc}$ at most $k$-times, therefore every trace of $\Sys$ under labeling map must be in the language of the automaton $\Aa_{c,k}$.
\qed
\end{proof}

To verify whether the traces of $\Sys$ under a labeling function $\Ll$ is accepted by a UCA $\Aa_{c} = (\Sigma,Q, Q_0, \delta, Q_{Acc})$, we search for a suitable $k$ such that the traces of $\Sys$ are in the language of the $k$-UCA $\Aa_{c,k} = (\Sigma,Q, Q_0, \delta, Q_{Acc})$ following Corollary~\ref{cor:kUCA_UCA}.
This is then cast as searching for a suitable CBBC for the product.
If one is successful in finding such a CBBC for some $k \in \N$, then we infer the traces of the system to be in the language of the automaton.
Here, the value $k$ is initially set to $0$, and is increased incrementally by one, every time we fail to find a CBBC until a suitable threshold is reached.
If a CBBC is found before the threshold is reached we stop our search and guarantee that the traces of the system are in the language of $\Aa_c$ and if not our procedure is inconclusive. 
Note that one \emph{does not need} to take the product of the system and the automaton explicitly to check the satisfaction of conditions and can define functions that are piecewise over automata states similar to \cite{murali_2023_closure}.
To verify whether the traces of a system $\Sys$ under a labeling function $\Ll$ are not contained in the language of an NBA $\Aa_b = (\Sigma,Q, Q_0, \delta, Q_{Acc})$, 
We first determine whether the traces of the system $\Sys$ are contained in the language of the UCA $\Aa_c =(\Sigma,Q, Q_0, \delta, Q_{Acc})$ as specified in the previous paragraph.
If so, then no trace of the system is in the language of the NBA $\Aa_b$ and if not our procedure is inconclusive.

\subsection{Limitations of CBBCs}
\VM{
While, we have shown that CBBCs are not only sufficient, but also necessary conditions to verify $k$-UCA properties they unfortunately, do not lend themselves readily to synthesis.
The key challenge for synthesis of controllers against $\omega$-regular specifications is that one may not use nondeterministic or universal automata due to issues related to realizability stemming from the nondeterminism or universality \cite{pnueli1989synthesis,chatterjee2013automata}.
Thus, one must instead use deterministic variants of these automata, however deterministic B\"uchi automata cannot capture all $\omega$-regular specifications.
Instead one must use deterministic $\omega$-automata such as parity or rabin automata \cite{baier_2008_principles}. However, the accepting conditions of these automata need to show not only that a region is visited finitely often but also that some regions are visited infinitely often.
A direct application of CBBCs to show regions are visited infinitely often, leads to one having to deal with unbounded counters making them a poor choice for automated approaches.
However, CBBCs may provide an effective tool to synthesize controllers for classes of $\omega$-automata such as as good-for-games automata \cite{henzinger2006solving}.
We leave an investigation of this as future work.
}

\section{CBBC Computation}
\label{sec:comp}
We present two approaches to compute co-B\"uchi barrier certificates as in Definition~\ref{def:cbar} .
The first relies on using Satisfiability Modulo Theory (SMT) solvers, while the second is based on a sum-of-squares (SOS) approach.
To find a CBBC as in Definition~\ref{def:cbar}, we first fix the template of the CBBC to be a piece-wise polynomials of a fixed degree $d$ with unknown coefficients $c_1, \ldots, c_z$, \textit{i.e.} $\Bb_i(x) = \sum\limits_{m = 1}^{z} c_{m,i} p_{m,i}(x)$ where the functions $p_{m}$ are monomials over the states $x$ and function $\Bb_i$ are defined for each counter value $i \in [0..k]$.
As $i \in [0..k]$ takes only finitely many values, we consider $\Bb$ to be a piecewise function with $k$ pieces.
We now discuss how we make use of a sum-of-squares approach to search for CBBCs.

\subsection{Sum-of-Squares based Approach}
One may make use of a semidefinite programming approach~\cite{Parrilo_2003} via sum-of-squares (SOS) similar to a search for a standard barrier certificates~\cite{prajna_2004_safety}.
A  set $A \subseteq \R^n$ is semi-algebraic if it  can be defined with the help of a vector of polynomial inequailites $h(x)$ as $A = \{ x \mid h(x) \geq 0 \}$, where the inequalities is interpreted component-wise.

To adopt a SOS approach to find CBBCs as in Definition~\ref{def:cbar}, we consider the sets $\Xx$, $\Xx_0$, and $\Xx_{VF}$ to be semi-algebraic sets defined with the help of vectors of polynomial inequalities $g$, $g_0$, and $g_{VF}$ respectively.
As these sets are semi-algebraic, the set $\Xx \setminus \Xx_{VF}$ is semi-algebraic as well, and let this set be defined by polynomial inequalities $g_N$.
Furthermore, we assume that the function $f$ is polynomial.
Then the search for a CBBC reduces to showing that the following polynomials are sum-of-squares for all $i \in [0..k]$, and $j \in [0..(k-1)]$:

\begin{align}
  & -\Bb_{0}(x)-\lambda_{0}^T(x)g_{0}(x),  \label{eq:sos_cbar_1} \\
  & \Bb_{k}(x) - \lambda_{VF}^T(x)g_{VF}(x) - \varepsilon \label{eq:sos_cbar_3} \\
  & -\Bb_{i}( f(x)) + \Bb_{i}(x) - \lambda_{N,i}^T(x)g_{N}(x), \text{ and }  \label{eq:sos_cbar_4}\\
  & -\Bb_{(j+1)}( f(x)) + \Bb_{(j)}(x) - \lambda_{V,i}(x)^Tg_{VF}(x), \label{eq:sos_cbar_5}
\end{align}
where the multipliers $\lambda_{0}$, $\lambda_{VF}$, $\lambda_{N,i}$, and $\lambda_{V, i}$ are sum-of-squares over $\Xx$, and $\varepsilon \in\R_{ > 0}$ is some small positive value.

\begin{lemma}
Assume the sets $\Xx$, $\Xx_{0}$, and $\Xx_{VF}$ are semi-algebraic, and there exists $k$ sum-of-squares polynomials $\Bb_\ell(x)$  satisfying conditions~\eqref{eq:sos_cbar_1} to~\eqref{eq:sos_cbar_5} for all $0 \leq \ell \leq k$.
Then the function $\Bb(x,i) = \Bb_i(x)$ is a CBBC  satisfying conditions~\eqref{eq:cbar_cond_1} to~\eqref{eq:cbar_cond_5}.
\end{lemma}

To find an SOS formulation for the product, first observe that there are finitely many letters $\sigma \in \Sigma$, and so without loss of generality, one can partition the set $\Xx$ into finitely many partitions $\Xx_{\sigma_1}, \ldots, \Xx_{\sigma_p}$ where for all $x \in \Xx_{\sigma_m}$, we have $\Ll(x) = \sigma_m$.
Given an element $\sigma_m \in \Sigma$, we can uniquely characterize the relation $\delta_{\sigma_m}$ as $(\ell, \ell') \in \delta_{\sigma_m}$ if and only if  $\ell' \in \delta(\ell, \sigma_m)$.
Assume that the sets $\Xx$, $\Xx_{0}$, and $\Xx_{\sigma_m}$ for all $\sigma_m$ are semi-algebraic and characterized by polynomial vector of inequalities $g(x)$, $ g_0(x)$, and $g_{\sigma_m}(x)$, respectively. Furthermore, assume that the state transition function $f:\Xx \to \Xx$ is polynomial. 
Then similar to the conditions above, one can reduce a search for CBBCs to checking if relevant polynomials are SOS.
To do so, we assume that the CBBC $\Bb(x,\ell,i) $ is defined as piecewise functions $\Bb_{\ell, i}(x)$ for each state of the automaton $\ell \in Q$,as there are finitely many states in $Q$, and similarly, for a fixed $k$, $i$, and $j$ take on value from $[0..k]$ and $[0..(k-1)]$ repsectively.
 Now, one can reduce the search for a CBBC to showing that the following polynomials are SOS:
\begin{align}
  & -\Bb_{\ell,0}(x)-\lambda_{0}^T(x)g_0(x), \text{ for all }\ell \in Q_0,  \label{eq:sos_1} \\
  & \Bb_{\ell, k}(x) - \lambda^T_{1}(x)g(x) - \varepsilon, \text{ for all } \ell \in Q_{Acc}, \label{eq:sos_3} 
\end{align}
where $\varepsilon \in \R_{>0}$ is some small positive value,
and for all $\sigma_{m} \in \Sigma$, $i \in [0..k]$ and all $\ell \in Q \setminus Q_{Acc}$, $\ell' \in \delta(\ell, \sigma_m)$ the following polynomial is sum-of-squares:
\begin{align}
 & - \Bb_{\ell, i}(x) +   \Bb_{ \ell', i}(f(x)) - \lambda_{2,\sigma_m,\ell, i}^T(x) g_{\sigma_m}(x),  \label{eq:sos_4}
 \end{align}
and for all $\sigma_{m} \in \Sigma$, $j \in [0..(k-1)]$ and all $\ell \in Q \cap Q_{Acc}$, $\ell' \in \delta(\ell, \sigma_m)$ the following polynomial is sum-of-squares:
 \begin{align}
  & - \Bb_{\ell,j}(x) +   \Bb_{\ell, (j+ 1)}(f(x)) - \lambda_{3,\sigma_m,\ell, i}^T(x) g_{\sigma_m}\!(x), \label{eq:sos_5}
\end{align}
where the multipliers $\lambda_0, \lambda_1$, $\lambda_{2,\sigma_m,\ell, i}$ $\lambda_{3,\sigma_m,\ell, i}$ are all sum-of-squares polynomials over the states $\Xx$.
Note we enumerate the conditions for every  $\sigma_m \in \Sigma$ and explicitly specify the values $\ell'$ based on the relation $\delta_{\sigma_m}$.

\begin{lemma}
Assume the sets $\Xx$, $\Xx_{0}$, $\Xx_{VF}$, and $\Xx_{\sigma_i}$ for all $\sigma_i$ are semi-algebraic, and there exists sum-of-squares polynomials $\Bb_{\xi, \ell}(x)$ satisfying conditions~\eqref{eq:sos_1} to~\eqref{eq:sos_5} for all $\xi \in Q$, and $\ell \in [0..k]$.
Then the function $\Bb(x, \xi, \ell) = \Bb_{\xi, \ell}(x)$ is a CBBC for the product satisfying conditions~\eqref{eq:prod_cond_1} to~\eqref{eq:prod_cond_5}.
\end{lemma}

One can then encode conditions \eqref{eq:sos_cbar_1} to~\eqref{eq:sos_cbar_5} as well as conditions~\eqref{eq:sos_1} to~\eqref{eq:sos_5} in a suitable SOS solver such as TSSOS~\cite{wang2021tssos} by fixing the barrier certificate to be a polynomial of degree $d$ and checking if the above constraints are satisfiable.

\VM{
The time complexity of checking if a polynomial in $n$ variables of degree $d$ is SOS involves solving a semidefinite feasibility problem in $\binom{n+d}{d}$ variables \cite{Parrilo_2003}. 
For the above approach, we have $O(k)$ polynomial constraints, as we can define function $\Bb$ to be a set of $k$-functions that are piecewise over the values of $k$.
Thus, computing a CBBC is polynomial in $O(k\binom{n+d}{d})$.
When dealing with the product of the system and a $k$-UCA $\Aa = (\Sigma,Q,Q_0,\delta, Q_{Acc})$, this is now polynomial in $O(k|Q|\binom{n+d}{d})$ variables.
}
\VM{We summarize the Algorithm for the SOS-based approach to verify UCA specifications in Algorithm \ref{alg:CBBC}.}

\begin{algorithm}
\VM{
\caption{CBBC computation for verifying $k$-UCA specifications}
\label{alg:CBBC}
\begin{algorithmic}[H]
\Procedure{CBBC\_Comp($\Sys$, $\Aa_c$, $\Ll$, $k_{\max}$, $d$) }{} 
\noindent \State{\textbf{Inputs}: System $\Sys = (\Xx, \Xx_0,f)$,}
\State{\qquad \qquad UCA $\Aa_{c} = (\Sigma, Q, Q_0, \delta, Q_{Acc})$,}
\State{\qquad \qquad Labeling map $\Ll: \Xx \to \Sigma$,} 
\State{\qquad \qquad Maximum counter threshold $k_{\max}$,}
\State{\qquad \qquad Degree of the certificates $d$.}
\State{$k := 0$}
\While{$k \leq k_{max}$  }
\State{Encode SOS polynomials from conditions \eqref{eq:sos_1}-\eqref{eq:sos_5} as an SDP following \cite{prajna_2004_safety}.}
\If{ SDP is feasible }
\State{ \textbf{return} $(status, \{\Bb_{i,\ell} \mid i \in Q, \ell \in [0..k] \})$}
\Else
\State{$k := k + 1$}
\EndIf
\EndWhile
\State{ \textbf{return} inconclusive}
\EndProcedure
\end{algorithmic}}
\end{algorithm}

\subsection{SMT-based Approach}
Alternative to SOS, one may make use of a Counterexample-guided Inductive Synthesis (CEGIS)~\cite{Lezama_2008_thesis} approach to find these CBBCs.
Observe that if the value of $x \in \Xx$ and $i \in \N$ is fixed, then the only decision variables in $\Bb(x,i)$ are the coefficients $c_m$.
We  sample $N$ points from the state space of the system to create the set $S = \{ x_1, \ldots, x_N \}$, and encode the constraints of the barrier certificate for each point $x_{j} \in S$ and all $0 \leq i \leq k$ as an SMT-query over the theory of linear real arithmetic (LRA)~\cite{dutertre_2006_fast} using z3~\cite{moura_2008_z3}.
We then find values $c_1, \ldots, c_z $  for the coefficients and substitute them as a candidate CBBC $\Bb(x,i)$.
To determine if this candidate is a CBBC, we now try to find a $x \in \Xx$ such that one of the conditions~\eqref{eq:cbar_cond_1} to~\eqref{eq:cbar_cond_5} does not hold.
We do this by encoding the negation of these conditions as an SMT query.
If such a counterexample $x_{counter}$ is found, we add it to the set $S$ and repeat the process.
If no such counterexample is found, then we conclude that this function is indeed a CBBC.
We adopt a similar approach to find a CBBC for the synchronized product, that acts as a proof that the traces of the system are contained in the language of a $k$-UCA $\Aa_{c,k} = (\Sigma, Q, Q_0, \delta, Q_{Acc})$.
Formally, we sample $N$ points from the state set $\Xx$, and check the conjunction of the following queries for every $ \ell\in Q$ : 
\begin{align}
&\ \underset{x_j \in \Xx_0, \ell\in Q_0}{\bigwedge} \Big(  \Bb_{\ell, 0}(x_j) \leq 0 \Big), && \label{eq:smt_cprod_1} \\
& \quad  \underset{x_j \in \Xx, \ell \in Q_{Acc} }{\bigwedge} \Big(  \Bb_{\ell, k}(x_j) > 0 \Big), &&\label{eq:smt_cprod_3} 
\end{align}
and for all states $\ell \in [0..(|Q|-1)]$ such that $\ell' \in \delta(q_\ell, \Ll(x_i) )$, and for all counter values $i \in [0..k]$, we have 
\begin{align}
    &\ \underset{x_j \in \Xx_0, \ell \in (Q \setminus Q_{Acc} )}{\bigwedge} \Big(  \Bb_{\ell',i}(f(x_j)) \leq \Bb_{\ell, i}(x_j)  \Big), \label{eq:smt_cprod_4} \\
    &\ \underset{x_j \in \Xx_0, \ell \in Q_{Acc} }{\bigwedge} \Big(  \Bb_{\ell', (i+1)}(f(x_j)) \leq \Bb_{\ell, i}(x_j)  \Big). \label{eq:smt_cprod_5} 
\end{align}

\VM{Unlike SOS, the above approaches are NP-hard~\cite{gao_dreal_2013}, and there is no guarantee that our CEGIS approach will terminate.
}
We have shown the usefulness of CBBCs in Theorem~\ref{thm_fin_visit} and Lemma \ref{lem:rel_comp_cbbc}. Our first  question was to investigate whether classical approaches for computing barrier certificates can be used for CBBCs.
We now show that our approach generalizes the state triplet approaches presented in~\cite{wongpiromsarn_2015_automata}.
\VM{Intuitively, our goal is as follows: we wish to show that if one is able to find a proof that a system $\Sys$ satisfies a desired $\omega$-regular property via NBA $\Aa_b$ through the state-triplet approach, then one is also able to find a proof that $\Sys$ satisfies the same property via CBBCs.
}

\section{Subsuming the state triplet approach}
\label{sec:subsumption}
\VM{Our result here follows from a simple observation that there exist systems which cannot be proved to verify UCA properties via the state-triplet approach, however as CBBCs are not only sufficient but also necessary, one can find such certificates for these systems.
Furthermore, while the above summary is non-constructive, we provide a constructive proof in Appendix \ref{subsec:construct_subsump}, which shows that the degree of our CBBCs are at most the same degree as the barrier certificates obtained from the triplet approach.
}
A similar discussion of such proofs also appears in \cite{murali_2023_closure} in the context of closure certificates and transition invariants rather than barrier certificates.

\subsection{The State Triplet Approach}
\label{subsec:triplet}
Consider a system $\Sys = ( \Xx, \Xx_0, f)$, an NBA $\Aa_b = (\Sigma, Q, Q_0, \delta, Q_{Acc})$, and a labeling function $\Ll: \Xx \to \Sigma$.
We first unroll the simple cycles of the automaton $\Aa_b$ that start from an accepting state to construct an automaton $\Aa'_b = (\Sigma, Q', Q_0, \delta', Q'_{Acc})$ such that $L(\Aa'_b) = L(\Aa_b)$.  
The states $Q'$ of $\Aa'_b$ contain all the states $Q$ of $\Aa_b$ as well as one new state $q'_i$ for every state $q_i \in Q$ that can be reached from an accepting state.
Determining which states are reachable from an accepting state can be done in polynomial time~\cite{kozen_2006_theory} in the number of states of the automaton.
Intuitively, these new states are meant to determine the state runs of the system which have reached an accepting state and continue onward.
The set of initial states is  the same as $\Aa_b$, while the set of accepting states consist of the states $ q'_{Acc}$, \textit{i.e.}, those accepting states that can be reached from an accepting state.
The transitions $\delta'$ are the same as the transitions for $\Aa_{b}$, except for the transitions that leave an accepting state $q_{Acc} \in Q_{Acc}$ and additional transitions added.
The  transitions added are due to two reasons :
\begin{itemize}
    \item For every pair of states $q_i, q_j \in Q$ that are reachable from an accepting state, and for any $\sigma \in \Sigma$ such that $q_j \in \delta(q_i,\sigma) $, we add a transition $q'_j \in \delta(q'_i, \sigma)$.
    This preserves all the transitions for the newly added states so they behave the same as before. 
    \item For every state $q_i \in Q$, and every accepting state $q_{Acc} \in Q_{Acc}$, and any letter $\sigma \in \Sigma$ such that $q_j \in \delta(q_{Acc},\sigma) $, we add a transition $q'_j \in \delta(q_{Acc}, \sigma)$. For every edge that leaves the accepting state, we add an edge to move to the states of the form $q' \in Q' \setminus Q$. Observe that there are no transition from the states in $Q' \setminus Q$ to the ones in $Q$.
\end{itemize}
To ensure that $T(\Sys, \Ll) \cap L(\Aa_{b}) = \emptyset$, it is sufficient to ensure that there exists a barrier certificate between some edge pair (or a state triplet) for every simple path from an initial state $q_0 \in Q_0 $ to some accepting state $q'_{Acc} \in Q'_{Acc}$ in the automaton $\Aa'_b$.
Note that unrolling the cycles  more than once does not change the edge pairs that enter or exit a state, and thus does not change the labels of the state triplets (edge pairs) that are encountered along a simple path from an initial state to an accepting state.
One may perform the same construction on the automaton again, however this does not add any new state triplets to consider.
The proof for this appears in \cite[Lemma 9]{murali_2023_closure}.
Thus, one is unable to verify those traces of a system which reach and cycle on an accepting state more than twice, even if they visit an accepting state only finitely often.
Note that unrolling once does have an impact since we can now consider those state triplets along the simple paths from an accepting state to an accepting state.

To illustrate the state triplet approach, consider the NBA $\Aa_{b_1} = ( \Sigma,Q, Q_0, \delta, Q_{Acc})$ in Figure~\ref{fig:aut_eg_1}, discussed in the introduction, with $\Sigma = \{a, b\}$, $Q = \{q_0, q_1, q_2\}$, $Q_0 = \{q_0 \}$, $Q_{Acc} = \{ q_2 \}$ and $\delta$ specified by the edges in the graph.
This NBA accepts those words which start with $b$ and have at least one $a$ in them and can be used to verify the safety of the system.
We, unfold $\Aa_{b_1}$ to get the NBA $\Aa'_{b_1} = (\Sigma', Q', Q_0, \delta, Q'_{Acc})$ as shown in Figure~\ref{fig:aut_eg_4}.

\begin{figure}[t]
    \centering
    \begin{tikzpicture}[node distance =2cm]
    \node[initial, state, draw, initial text =,fill=blue!10!white] (0) at (0,0) {$q_0$};\
    \node[,state, fill=blue!10!white,] (1) at (2,0) {$q_1$};
    \node[ ,state, fill=blue!10!white,] (2) at (4,0) {$q_2$};
    \node[accepting, state, fill = blue!10!white,] (3) at (6,0) {$q'_2$};
    \path[->]
    (0) edge node[above]{$b$} (1)
    (1) edge[loop below] node{$b$} (1)
    (1) edge node[above]{$a$} (2)
    (2) edge node[above]{$a,b$} (3)
    (3) edge[loop right] node{$a , b$} (3);
    \end{tikzpicture}
    \caption{ Automaton $\Aa'_{b_1}$ created from unfolding NBA $\Aa_b$ in Section~\ref{sec:intro} (Figure~\ref{fig:aut_eg_1}). }
    \label{fig:aut_eg_4}
\end{figure}

Consider a system $\Sys = (\Xx, \Xx_0, f)$, under a labeling map $\Ll$ which naturally partitions $\Xx$ into two sets $\Xx_{a}$ and $\Xx_b$.
We break the simple paths of the above automaton into state triplets or edge pairs and then find barrier certificates over these edge pairs. In this case there is a single simple path from the initial state $q_0$ to an accepting state $q'_2$ consisting of the states $(q_0, q_1, q_2, q'_2)$ and leads to two triplets $(q_0, q_1, q_2)$ and $(q_1, q_2, q'_2)$.
We take the first triplet and note this corresponds to the edges with labels $a$ and $b$.
We then search for a barrier certificate by considering the initial set of states to be the set $\Xx_{b}$, the set of states with a label $b$, and the unsafe set of states $\Xx_u$ to be the set $\Xx_{a}$, all the states with a label $a$.
If we are successful in finding such a barrier certificate, then we infer that we may cut this path and observe that no other path from the initial state $q_0$ to the accepting state $q'_2$ exists.
In such a case we verify that $T(\Sys, \Ll) \cap L(\Aa_b) = \emptyset$ and so that no trace of the system is in the language of NBA $\Aa_b$.
If we cannot find such a barrier certificate, we now search for a barrier certificate in the next triplet, the triplet $(q_0, q_1, q_2)$.
Observe that the edges $(q_1, q_2)$ and $(q_2, q'_2)$ share the label $a$, and as such one cannot simultaneously satisfy condition~\eqref{eq:bar_cond_1} and~\eqref{eq:bar_cond_2} on the states with the label $a$.
When this fails, our approach is inconclusive and we can say nothing about whether the system satisfies the property or not.

The state triplet approach is conservative in the following direction: Independent of the state runs in the automaton and of the initial states of the system $\Sys$, 
one is required to break the edge pairs of every simple path regardless of what states of the automaton may be encountered before or after.

\begin{theorem}
\label{thm:subsumption}
    Let $\Sys = (\Xx, \Xx_0, f)$ be a system, $\Ll:\Xx \to \Sigma$ be a labeling map, and $\Aa_b = (\Sigma, Q, Q_0, \delta, Q_{Acc})$ be an NBA. Suppose there exist barrier certificates $\Bb_1, \ldots, \Bb_m$ that act as a proof that no trace of the system is in $L(\Aa_b)$ by the state triplet approach, then there exists a CBBC $\Bb$ that also serves as a witness that no trace is in $L(\Aa_b)$.
\end{theorem}

We present a constructive proof of Theorem \ref{thm:subsumption} in Appendix \ref{subsec:construct_subsump}.
A nonconstructive proof follows from the fact that there exist systems which cannot be verified via the state triplet approach but can be via CBBCs.

\section{Case Studies}
\label{sec:case_studies}
The first case study demonstrates that while the state triplet approach is inadequate in demonstrating the satisfaction of the property, our approach finds CBBCs to verify the satisfaction of the property.
We then use the SOS-based approach to find certificates for higher dimensional oscillator models to demonstrate their effectiveness.

\subsection{Case Study Against the State Triplet Approach}
\label{subsec:case_study_triplet}
Consider a room-temperature model adapted from~\cite{anand_2022_k}, whose dynamics is specified as:
\[ \Sys: x(t+1) = (1 - \tau_s \alpha)x(t) + \tau_s\alpha T_e, \]
To demonstrate the effectiveness of our approach against properties specified by automata, we consider the same system $\Sys = (\Xx, \Xx_0, f)$ as in the previous Subsection,  with $\Xx = [17, 40] $ as the state set, $\Xx_0 = [30,35]$ as the initial set of states, and $f(x) = (1 - \tau_s \alpha)x + \tau_s\alpha T_e$
denotes the state transition function, where $\tau_s = 5$ minutes denotes the sampling time, $\alpha = 0.08$ the heat-exchange coefficient and $T_e = 17$C denotes the ambient temperature.
Consider the UCA $\Aa_c = (\Sigma, Q, Q_0, \delta, Q_{Acc})$ specified in Figure~\ref{fig:aut_1_case_study}, where $Q = \{ q_0, q_1\}$ denotes the set of states, $Q_0 = \{ q_0 \}$ the initial state, $Q_{Acc} = \{ q_1 \}$ the accept state, and the transition function $\delta$  specified by the edges in the graph.
\begin{figure}[b]
    \centering
    \begin{tikzpicture}[node distance =2cm]
    \node[initial, state, draw, initial text =,fill=blue!10!white] (0) at (0,0) {$q_0$};
    \node[accepting ,state, fill=blue!10!white,] (1) at (2,0) {$q_1$};
    \path[->]
    (0) edge[bend left] node[above]{$b$} (1)
    (1) edge[bend left] node[below]{$a,b,c$} (0)
    (0) edge[loop above] node{$a , c$} (0);
    \end{tikzpicture}
    \caption{ UCA $\Aa_c$ for Subsection~\ref{subsec:case_study_triplet}.}
    \label{fig:aut_1_case_study}
\end{figure}
Let the labeling function $\Ll: \Xx \to \Sigma$ be defined as:
\begin{align*}
        & \Ll(x) = a,& & \text{ for all } x \in (28,40] , \\
    & \Ll(x) = b, && \text{ for all } x \in [25,28] , \text{ and } \\
    & \Ll(x) = c && \text{ for all x } \in [17,25).     
\end{align*}

First, we observe that we cannot make use of the state triplet approach.
The initial set of states have a label $a$ and can reach states with a label $b$, thus one cannot cut the edge pair $(q_0, q_0)$ and $(q_0, q_1)$.
Similarly, one cannot cut the edge pair from $q_1$ to $q_0$ as the edge pairs $(q_0, q_1)$ and $(q_1,q_0)$ as well as the pair $(q_1, q_0)$, and $(q_0, q_0)$ have overlapping labels.
We thus make use of CBBCs to guarantee that the accepting state $q_1$ is visited finitely often.

We fix the value of $k$ to be $4$, and consider a degree-3 polynomial $\Bb$ in $x$ (i.e. the state of the system), $i$ (i.e. the state of the automaton), and $\ell$ (i.e. the counter value).
We then adopt a CEGIS approach to find such a CBBC, and found the function $\Bb(x, i, \ell) = -0.000311 \ell^3 -0.01770 x i \ell + 0.0727 \ell^2 +0.0442 x \ell +0.30096 i \ell - \ell - 0.001 $.
The time taken to find this CBBC was $1.5$ hours on the reference machine.

\subsection{Verification for a 2D Van der Pol oscillator model}
\label{subsec:case_study_osc}
\VM{
Next, we consider a 2D van-der Pol Oscillator, whose model dynamics is given by the following difference equations:
\begin{align}
&\Sys\begin{cases}
    x(t+1) = x(t) + T y(t),\\
    y(t+1) = y(t) + T(-x(t) + u y(t)(1-x(t)^2)), \nonumber
\end{cases}
\end{align}
where $T = 0.1$ is the sampling time and $u = 0.4$.
The state set, and initial set, are given by $\Xx = [-2.5,4]\times [-2.5,2.5]$, $\Xx_0 = [3,3.5]\times[1.5,2]$respectively. 
We consider the UCA $\Aa_{c}$ in Figure \ref{fig:aut_2_case_study}, with \[\Ll(x,y) = \begin{cases}
a & \text{ if }  x \in [3,4], y \in [-1,0.5], \\
c & \text{ if }  x \in [2,4], y \in [-2.5,-2], \\
b & \text{ otherwise.}
\end{cases} \]
\begin{figure}[b]
    \centering
    \begin{tikzpicture}[node distance =2cm]
    \node[initial, state, draw, initial text =,fill=blue!10!white] (0) at (0,0) {$q_0$};
    \node[accepting ,state, fill=blue!10!white,] (1) at (2,0) {$q_1$};
    \node[state, accepting, fill = blue!10!white] (2) at (1,-2) {$q_2$};
    \path[->]
    (0) edge[bend left] node[above]{$a$} (1)
    (1) edge[bend left] node[below]{$b$} (0)
    (1) edge[bend left] node[below]{$c$} (2)
    (1) edge[bend left] node[below]{$c$} (2)
    (0) edge[bend right] node[left]{$c$} (2)
    (0) edge[loop above] node{$b$} (0)
    (1) edge[loop above] node{$a$} (1)
    (2) edge[loop below] node{$a,b,c$} (2);
    \end{tikzpicture}
    \caption{ UCA $\Aa_c$ for Subsection~\ref{subsec:case_study_osc}.}
    \label{fig:aut_2_case_study}
\end{figure}
We then search for a CBBC as in Definition \ref{def:cbar_prod} following Algorithm \ref{alg:CBBC} by using the TSSOS \cite{wang2021tssos} package by setting $k_{\max} = 8$, and incrementing the value of $d$ from $d=1$ to $d = 8$.
We are able to find a CBBC of degree $6$ with $k = 3$, ensuring that the accepting states $q_1$, and $q_2$ are visited at most $3$ times from an initial state of the system.
The time taken for this is around $5$ minutes on the reference machine and the certificates are described in Appendix \ref{ap:coeffs_osc}.
}
\subsection{Verification of a 3D Kuramoto Oscillator Model}
\label{subec:case_study_Kuramoto}
\VM{
As a final case study, we consider a three dimensional Kuramoto oscillator where the dynamics are given by:
\begin{align}
\Sys
&\begin{cases}
    x(t+1) &= x(t) + TKsin(y(t)-x(t))  \\ & \qquad + cx(t)^2 + \Omega,\\
    y(t+1) &= y(t) + TK\big(sin(x(t)-y(t)) \\&\qquad+ sin(z(t)-y(t)) \big) +  \Omega +cy(t)^2,\\
    z(t+1) &= z(t) + TKsin(y(t)-z(t)) \\ & \qquad +cz(t)^2 + \Omega, \nonumber
\end{cases}
\end{align}
where $T = 0.1s$ is the sampling time, $\Omega = 1$ is the natural frequency added to an offset, $K = 0.06$ is the coupling strength, and $c = -0.532$. 
As the above dynamics uses a trigonometric function, we approximate it using a Taylor's approximation as $sin(r) \approx r- \frac{r^3}{6}$ for all the corresponding sine terms.
We consider $\Xx = [0,2]^3$ to be the set of states, $\Xx_{0} = [0, \frac{2\pi}{15}]^3$ to be the initial set of states, and $\Xx_{VF} = [0,0.7]^2 \times [0,2]$ to be the set of states that must be visited only finitely often.
We then search for a CBBC as in Definition \ref{def:cbar_prod} following Algorithm \ref{alg:CBBC} by using the TSSOS \cite{wang2021tssos} package by setting $k_{\max} = 8$, and incrementing the value of $d$ from $d=1$ to $d = 8$.
We find a CBBC as in Definition \ref{def:cbar}, of degree $d = 4$ with $k = 2$.
The time taken for this is around $20$ minutes and the certificates are given in Appendix \ref{ap:coeffs_kuram}.
}

\section{Conclusion}
\label{sec:concl}
This paper advocated the notion of co-B\"uchi barrier certificates that act as a proof that a region is visited only finitely often.
These CBBCs have been demonstrated to be useful in the automata-theoretic verification of discrete-time dynamical systems.
This paper shows that techniques to compute standard barrier certificates can be extended to search for CBBCs.
Some promising next steps are to understand the applicability of CBBCs in verifying quantitative and time-critical properties of dynamical systems via weighted, stochastic, and timed automata based specifications.

\bibliographystyle{elsarticle-num}
\bibliography{bibliography}

\appendix
\section{A Constructive Proof for subsumption of the state-triplet approach}
\subsection{Our Construction}
\label{subsec:construct_subsump}
We now show that our approach generalizes the earlier state triplet approach.
Consider the NBA $\Aa'_b = (\Sigma, Q', Q_0, \delta', Q_{Acc})$, and let us assume that there exist $m$ barrier certificates $\Bb_1, \Bb_2 \ldots, \Bb_m$ associated with state triplets $(q_i, q'_i, q''_i)$, (or edge pairs $(e_i, e'_i)$) for each $1 \leq i \leq m$, such that $q'_i \in \delta'(q_i, a_i)$ and $q''_i \in \delta'(q'_i, b_i)$ that act as a proof certificate that no trace of the system is in $L(\Aa'_b)$.

First, for convenience we assume that there are no edge pairs $(e_i, e'_i)$ and $(e_j, e'_j)$ for all $1 \leq i, j \leq m$, $i \neq j$,  such that the edges $e_i$ and $e_j$ are incoming edges to the same state and the edges $e'_i$ and $e'_j$ are outgoing edges from the same state, \textit{i.e.}, there are no two distinct state triplets $(q_i, q'_i, q''_i)$ and $(q_j, q'_j, q''_j)$ which share a common middle element ($q'_i \neq q'_j$ for all $0 \leq i,j \leq m$, $i \neq j$).
Situations which may arise are based on the example shown in Figure~\ref{fig:edge_cases}.
To consider a single triplet, we replace their barrier certificates, with a new one by taking either the minimum or maximum  of the barrier certificates $\Bb_i$ and $\Bb_j$ or neither. 
We list three different cases as follows:
\begin{itemize}
    \item Suppose we have the edges $e_i$ and $e_j$ to be the same edge and the edges $e'_i$ and $e'_j$ to be different. 
    In this case we have $a_i = a_j$.
     We know that for any $x\in \Xx$ with $\Ll(x) = a_i$, we have both $\Bb_i(x) \leq 0$ and $\Bb_j(x) \leq 0$, and at least one of $\Bb_i(x) > 0$ or $\Bb_j(x) > 0$ for all states $x \in \Xx$ such that either $\Ll(x) = b_i$ or $\Ll(x) = b_j$. We therefore consider $\Bb'(x) = \max\{\Bb_i(x), \Bb_j(x)\}$ as the new barrier certificate for both the edge pairs.
    \item Suppose we have the edges $e'_i$ and $e'_j$ to be the same edge and the edges $e_i$ and $e_j$ to be different.
    In this case we have $b_i = b_j$.
    We know that for any $x\in \Xx$ with $\Ll(x) = a_i$, we have at least one of $\Bb_i(x) \leq 0$ or $\Bb_j(x) \leq 0$, and both $\Bb_i(x) > 0$ and $\Bb_j(x) > 0$ for all states $x \in \Xx$  with $\Ll(x) = b_i$. We therefore consider $\Bb'(x) = \min\{\Bb_i(x), \Bb_j(x)\}$ as the new barrier certificate for both the edge pairs.
    \item Neither edges are the same but they are incident on the same vertex, \textit{i.e.}, the states $q'_i$ and $q'_j$ are the same for both of the triplets.
    In this case we do not consider the barrier certificate as both paths continue on and must be cut at some future point, namely there is no barrier certificate between $e_i$ and $e'_j$ nor is there one between $e_j$ and $e'_i$ and so there must be some future state triplet or edge pairs which disallow paths starting from $e'_i$ and $e'_j$.
\end{itemize}
\begin{figure}[b]
    \centering
    \begin{tikzpicture}[node distance =2cm]
    \node[ state, draw, initial text =,fill=blue!10!white] (0) at (1,1) {$q'$};
    \node[] (1) at (0,0) {};
    \node[] (2) at (2,2) {};
    \node[] (3) at (0,2) {};
    \node[] (4) at (2,0) {};
    \path[->]
    (1) edge node[above]{$e_i$} (0)
    (0) edge node[above] {$e'_i$} (2)
    (3) edge node[above]{$e_j$} (0)
    (0) edge node[above] {$e'_j$} (4);
    \end{tikzpicture}
    \caption{Corner cases where multiple barrier certificates are present for the same intermediate state. In all cases we can reduce this to a single barrier certificate over the state $q'$. Here we can either replace the barrier certificates with the maximum or minimum or ignore these. }
    \label{fig:edge_cases}
\end{figure}

Now, we divide the states of NBA $\Aa'_b$ into two sets $Q'_l$ or $Q'_r$ based on whether the state is to the left or the right of \emph{some} triplet.
If a state $q \in Q'$ is to the left of, or is, the middle element of some state triplet along a path from an initial state to an accepting state, then it is in  $Q'_l$, else it is in $Q'_r$.
Observe that all accepting states must be in $Q'_r$ as there is some triplet for every path by our assumption.

We now define a CBBC $\Bb$ as in Definition~\ref{def:cbar_prod}, with the counter value $k = 0$, over the synchronized product as follows.
First, for all states $x \in \Xx$ and every state $q_l \in Q'_l$, we define $\Bb( x, l, 0 )$ as the maximum value of the barriers $\Bb_1, \ldots, \Bb_m$ over the regions. For all states $x \in \Xx$ and $q_r \in Q'_r$, we define $\Bb( x, r, 0)$ and $\Bb( x, r, 1 )$ as the minimum value of the barrier certificates $\Bb_1, \ldots, \Bb_m$ over the regions $b_1, \ldots, b_m$. In summary, we get:

 \[\Bb( x,\xi,\ell ) =  
 \begin{cases}
\underset{1 \leq i \leq m}{\max} \big\{ {\underset{\underset{\Ll(x)= a_i}{x \in \Xx} }{\max} \{ \Bb_i(x) \} } \big\} \quad \text{for all } q_{\xi} \in Q'_l, 
\ell = 0\\

\underset{1 \leq i \leq m}{\min}\big \{ \underset{\underset{\Ll(x)= b_i}{x \in \Xx} }{\min} \{ \Bb_i(x) \} \big\}, \text{ for all } q_{\xi} \in Q'_r, \ell \in\{ 0,1\}.
 \end{cases}
 \]

Observe that the largest value this achieves for all states $q'_l \in Q'_l$ is still less than $0$ and so we satisfy condition~\eqref{eq:prod_cond_1}. 
Next,
the value of the CBBC for all states $q'_r \in Q'_r$ is greater than $0$ and so we satisfy condition~\eqref{eq:bar_cond_3} for all states $q'\in Q_{Acc}$.

Conditions~\eqref{eq:prod_cond_4} and~\eqref{eq:prod_cond_5} hold within the sets $Q'_l$ and $Q'_r$ as they are constants in the two partitions.
Secondly, we observe that no state in $Q'_r$ can be reached by any trace of the system as all the state triplets have been cut, and so there exists no state $x \in \Xx$ and $q_i \in Q'_l$ such that there exists a state $q'_j \in Q'_r$ with $q'_j \in \delta(q_i, \Ll(x))$.
This guarantees conditions~\eqref{eq:prod_cond_4} and~\eqref{eq:prod_cond_5} hold for every state $x \in \Xx$ and $q \in Q$.
We  illustrate this partition of the states of the NBA $\Aa'_{b}$  in Figure~\ref{fig:triplet_div}, where the states $q'_i$ denote the middle state of the triplets.
\begin{figure}[b]
    \centering
    \begin{tikzpicture}[node distance =2cm]
    \node[ initial, state, draw, initial text =,fill=blue!10!white] (0) at (0,0) {$q_{0,0}$};
    \node[] (1) at (0,-2) {$\ldots$};
    \node[] (2) at (0,-4) {$\ldots$};
    \node[initial, state, draw, initial text =,fill=blue!10!white] (3) at (0,-6) {$q_{0, \ell}$};
    \node[] (4) at (2,0) {$\ldots$};
    \node[] (5) at (2,-2) {$\ldots$};
    \node[] (6) at (2,-4) {$\ldots$};
    \node[] (7) at (2,-6) {$\ldots$};

    \node[, state, draw, initial text =,fill=blue!10!white] (8) at (4,0) {$q'_{1}$};
    \node[] (9) at (4,-2) {$\ldots$};
    \node[, state, draw, initial text =,fill=blue!10!white] (10) at (4,-4) {$q'_{m}$};
    \node[state, draw, accepting, fill=blue!10!white] (11) at (7,-2) {$q_{Acc}$};
    \node[] (12) at (5.5,0) {$\ldots$};
    \node[] (13) at (5.5,-2) {$\ldots$};
    \node[] (14) at (5.5,-4) {$\ldots$};
    \node[] (15) at (5.5,-6) {$\ldots$};
    \path[->]
    (0) edge node[above]{$\sigma_1$} (4)
    (3) edge node[below]{$\sigma_2$} (7)
    (4) edge node[above]{$a_1$} (8)
    (7) edge node[below, right]{$a_m$} (10)
    (8) edge node[above]{$b_1$} (13)
    (10) edge node[below]{$b_m$} (14)
    (14) edge node[below]{$\sigma_4$} (11)
    (13) edge node[above]{$\sigma_3$} (11);
    \end{tikzpicture}
    \caption{ The NBA $\Aa'_b$ is  divided along the middle elements of the state triplets. The states to the left of (and including) the states $q'_i$, for all $1 \leq i \leq m$,  are in $Q'_l$ and the right are in $Q'_r$.}
    \label{fig:triplet_div}
\end{figure}

Observe that the CBBC can be defined as a piecewise function which is fixed to a constant value less than $0$ for all states $q_l \in Q'_l$ and $x \in \Xx$, and to a constant value greater than $0$ for all states $q_r \in Q'_r$ and states $x \in \Xx$. 

\section{CBBC for Section \ref{subsec:case_study_osc}}
\label{ap:coeffs_osc}
The CBBC $\Bb_{\ell,i}(x)$ are defined picewise over $\ell \in Q$ ($\ell \in [0..2]$), and $i \in [0..3]$ as:
\begin{description}
\item $\Bb_{0,0}(x,y) = -0.001479 \cdot x^{6} - 0.012757 \cdot x^{5} \cdot y + 0.05985 \cdot x^{4} \cdot y^{2} - 0.075517 \cdot x^{3} \cdot y^{3} + 0.071967 \cdot x^{2} \cdot y^{4} - 0.042324 \cdot x \cdot y^{5} + 0.010623 \cdot y^{6} - 0.003189 \cdot x^{5} - 0.011026 \cdot x^{4} \cdot y - 0.004925 \cdot x^{3} \cdot y^{2} - 0.015355 \cdot x^{2} \cdot y^{3} + 0.00277 \cdot x \cdot y^{4} - 0.005104 \cdot y^{5} + 0.117479 \cdot x^{4} - 0.022629 \cdot x^{3} \cdot y - 0.151256 \cdot x^{2} \cdot y^{2} + 0.299684 \cdot x \cdot y^{3} - 0.037504 \cdot y^{4} + 0.017933 \cdot x^{3} + 0.061936 \cdot x^{2} \cdot y + 0.005285 \cdot x \cdot y^{2} + 0.030286 \cdot y^{3} - 1.727458 \cdot x^{2} + 0.393963 \cdot x \cdot y - 1.348065 \cdot y^{2} + 1.0e-6 \cdot x - 1.0e-6 \cdot y + 6.101348$ \\
\item $\Bb_{0,1}(x,y) := 0.014904 \cdot x^{6} - 0.001828 \cdot x^{5} \cdot y + 0.103405 \cdot x^{4} \cdot y^{2} - 0.079955 \cdot x^{3} \cdot y^{3} + 0.109044 \cdot x^{2} \cdot y^{4} - 0.047584 \cdot x \cdot y^{5} + 0.017908 \cdot y^{6} + 0.001017 \cdot x^{5} - 0.00713 \cdot x^{4} \cdot y + 0.005648 \cdot x^{3} \cdot y^{2} - 0.018626 \cdot x^{2} \cdot y^{3} + 0.010244 \cdot x \cdot y^{4} - 0.006376 \cdot y^{5} - 0.014842 \cdot x^{4} - 0.130578 \cdot x^{3} \cdot y - 0.369789 \cdot x^{2} \cdot y^{2} + 0.293624 \cdot x \cdot y^{3} - 0.110469 \cdot y^{4} - 0.004996 \cdot x^{3} + 0.044031 \cdot x^{2} \cdot y - 0.025907 \cdot x \cdot y^{2} + 0.035927 \cdot y^{3} - 1.60888 \cdot x^{2} + 0.682265 \cdot x \cdot y - 1.300488 \cdot y^{2} + 1.0e-6 \cdot x + 1.926749$ \\

\item $\Bb_{0,2}(x,y) := 0.020866 \cdot x^{6} + 0.0065 \cdot x^{5} \cdot y + 0.114056 \cdot x^{4} \cdot y^{2} - 0.072159 \cdot x^{3} \cdot y^{3} + 0.114243 \cdot x^{2} \cdot y^{4} - 0.044565 \cdot x \cdot y^{5} + 0.018515 \cdot y^{6} + 0.004629 \cdot x^{5} - 0.004803 \cdot x^{4} \cdot y + 0.013346 \cdot x^{3} \cdot y^{2} - 0.01922 \cdot x^{2} \cdot y^{3} + 0.013768 \cdot x \cdot y^{4} - 0.006355 \cdot y^{5} - 0.045795 \cdot x^{4} - 0.208793 \cdot x^{3} \cdot y - 0.383569 \cdot x^{2} \cdot y^{2} + 0.245491 \cdot x \cdot y^{3} - 0.106145 \cdot y^{4} - 0.024058 \cdot x^{3} + 0.031517 \cdot x^{2} \cdot y - 0.045583 \cdot x \cdot y^{2} + 0.036307 \cdot y^{3} - 1.673498 \cdot x^{2} + 0.866638 \cdot x \cdot y - 1.407303 \cdot y^{2} + 1.0e-6 \cdot x + 1.364599$ \\

\item $\Bb_{0,3}(x) := 0.014596 \cdot x^{6} + 0.014443 \cdot x^{5} \cdot y + 0.091242 \cdot x^{4} \cdot y^{2} - 0.05221 \cdot x^{3} \cdot y^{3} + 0.092222 \cdot x^{2} \cdot y^{4} - 0.034871 \cdot x \cdot y^{5} + 0.0139 \cdot y^{6} + 0.003293 \cdot x^{5} - 0.004216 \cdot x^{4} \cdot y + 0.009595 \cdot x^{3} \cdot y^{2} - 0.015104 \cdot x^{2} \cdot y^{3} + 0.010328 \cdot x \cdot y^{4} - 0.004994 \cdot y^{5} + 0.030812 \cdot x^{4} - 0.280669 \cdot x^{3} \cdot y - 0.216901 \cdot x^{2} \cdot y^{2} + 0.16154 \cdot x \cdot y^{3} - 0.044066 \cdot y^{4} - 0.017154 \cdot x^{3} + 0.026044 \cdot x^{2} \cdot y - 0.03371 \cdot x \cdot y^{2} + 0.029005 \cdot y^{3} - 1.956452 \cdot x^{2} + 1.019239 \cdot x \cdot y - 1.657867 \cdot y^{2} + 2.587095$ \\

\item $\Bb_{1,0}(x,y) := 1.169368 \cdot x^{6} - 0.802076 \cdot x^{5} \cdot y + 1.61694 \cdot x^{4} \cdot y^{2} - 0.806867 \cdot x^{3} \cdot y^{3} + 3.620299 \cdot x^{2} \cdot y^{4} - 1.427064 \cdot x \cdot y^{5} + 1.50265 \cdot y^{6} - 7.151424 \cdot x^{5} + 3.095272 \cdot x^{4} \cdot y - 8.119412 \cdot x^{3} \cdot y^{2} + 1.904013 \cdot x^{2} \cdot y^{3} - 15.292198 \cdot x \cdot y^{4} + 4.045883 \cdot y^{5} + 9.863108 \cdot x^{4} - 0.909527 \cdot x^{3} \cdot y + 9.274002 \cdot x^{2} \cdot y^{2} - 0.648806 \cdot x \cdot y^{3} + 8.101657 \cdot y^{4} + 3.193949 \cdot x^{3} - 2.674896 \cdot x^{2} \cdot y + 5.269647 \cdot x \cdot y^{2} - 2.075598 \cdot y^{3} - 0.355375 \cdot x^{2} + 0.248207 \cdot x \cdot y + 0.604798 \cdot y^{2} - 0.109638 \cdot x - 0.235628 \cdot y + 4.284889$ \\

\item $\Bb_{1,1}(x,y) := 1.079006 \cdot x^{6} - 0.514729 \cdot x^{5} \cdot y + 3.468024 \cdot x^{4} \cdot y^{2} + 3.107501 \cdot x^{3} \cdot y^{3} + 11.825979 \cdot x^{2} \cdot y^{4} + 9.156119 \cdot x \cdot y^{5} + 4.958857 \cdot y^{6} - 6.65663 \cdot x^{5} + 2.319631 \cdot x^{4} \cdot y - 11.363004 \cdot x^{3} \cdot y^{2} + 2.422806 \cdot x^{2} \cdot y^{3} - 5.018148 \cdot x \cdot y^{4} + 6.555657 \cdot y^{5} + 9.146765 \cdot x^{4} - 2.107326 \cdot x^{3} \cdot y + 4.292546 \cdot x^{2} \cdot y^{2} - 5.384374 \cdot x \cdot y^{3} + 4.530556 \cdot y^{4} + 3.30525 \cdot x^{3} - 1.60986 \cdot x^{2} \cdot y + 2.758191 \cdot x \cdot y^{2} - 2.695313 \cdot y^{3} - 0.477553 \cdot x^{2} + 0.939179 \cdot x \cdot y + 0.13928 \cdot y^{2} + 0.069863 \cdot x + 0.026055 \cdot y + 3.61205$\\

\item $\Bb_{1,2}(x,y) := 0.8925 \cdot x^{6} - 0.742628 \cdot x^{5} \cdot y + 3.15367 \cdot x^{4} \cdot y^{2} + 3.382412 \cdot x^{3} \cdot y^{3} + 12.385971 \cdot x^{2} \cdot y^{4} + 9.857827 \cdot x \cdot y^{5} + 5.144468 \cdot y^{6} - 5.872276 \cdot x^{5} + 2.954834 \cdot x^{4} \cdot y - 10.317188 \cdot x^{3} \cdot y^{2} + 2.308949 \cdot x^{2} \cdot y^{3} - 4.777581 \cdot x \cdot y^{4} + 6.499851 \cdot y^{5} + 8.932045 \cdot x^{4} - 1.896671 \cdot x^{3} \cdot y + 4.762227 \cdot x^{2} \cdot y^{2} - 5.233601 \cdot x \cdot y^{3} + 4.614206 \cdot y^{4} + 2.101369 \cdot x^{3} - 1.36296 \cdot x^{2} \cdot y + 2.356684 \cdot x \cdot y^{2} - 2.4981 \cdot y^{3} - 0.911951 \cdot x^{2} + 1.140637 \cdot x \cdot y - 0.025658 \cdot y^{2} + 0.023114 \cdot x + 0.047989 \cdot y + 4.729123$ \\

\item $\Bb_{1,3}(x,y) := 0.657481 \cdot x^{6} + 0.481838 \cdot x^{5} \cdot y + 2.914344 \cdot x^{4} \cdot y^{2} + 2.589529 \cdot x^{3} \cdot y^{3} + 5.201449 \cdot x^{2} \cdot y^{4} + 2.90154 \cdot x \cdot y^{5} + 2.630666 \cdot y^{6} - 4.665976 \cdot x^{5} - 1.729505 \cdot x^{4} \cdot y - 12.526088 \cdot x^{3} \cdot y^{2} - 5.730965 \cdot x^{2} \cdot y^{3} - 7.31455 \cdot x \cdot y^{4} + 7.158644 \cdot y^{5} + 10.201324 \cdot x^{4} - 1.176185 \cdot x^{3} \cdot y + 12.251171 \cdot x^{2} \cdot y^{2} + 3.196583 \cdot x \cdot y^{3} + 13.482927 \cdot y^{4} - 7.195209 \cdot x^{3} + 4.947589 \cdot x^{2} \cdot y - 2.009166 \cdot x \cdot y^{2} + 1.422074 \cdot y^{3} + 2.104938 \cdot x^{2} + 2.740789 \cdot x \cdot y + 2.251575 \cdot y^{2} - 0.352426 \cdot x + 0.93496 \cdot y + 6.437098$\\

\item $\Bb_{2,0}(x,y) := 6.647834 \cdot x^{6} - 23.026921 \cdot x^{5} \cdot y + 63.958752 \cdot x^{4} \cdot y^{2} - 47.22306 \cdot x^{3} \cdot y^{3} + 21.22077 \cdot x^{2} \cdot y^{4} - 16.818629 \cdot x \cdot y^{5} + 20.858006 \cdot y^{6} - 17.291677 \cdot x^{5} - 21.229591 \cdot x^{4} \cdot y - 31.082763 \cdot x^{3} \cdot y^{2} - 18.195494 \cdot x^{2} \cdot y^{3} - 25.085691 \cdot x \cdot y^{4} - 5.694037 \cdot y^{5} + 29.391701 \cdot x^{4} + 28.493682 \cdot x^{3} \cdot y + 35.15935 \cdot x^{2} \cdot y^{2} - 2.539602 \cdot x \cdot y^{3} + 64.361465 \cdot y^{4} + 5.216234 \cdot x^{3} + 6.999094 \cdot x^{2} \cdot y + 3.81613 \cdot x \cdot y^{2} + 6.50213 \cdot y^{3} + 13.044233 \cdot x^{2} + 2.41708 \cdot x \cdot y + 16.618755 \cdot y^{2} + 0.582513 \cdot x + 2.187375 \cdot y + 15.750095$\\

\item $\Bb_{2,1}(x,y) := 2.105094 \cdot x^{6} - 1.453849 \cdot x^{5} \cdot y + 35.178515 \cdot x^{4} \cdot y^{2} - 49.88313 \cdot x^{3} \cdot y^{3} + 45.685832 \cdot x^{2} \cdot y^{4} - 4.629385 \cdot x \cdot y^{5} + 9.011891 \cdot y^{6} - 13.536015 \cdot x^{5} - 24.339273 \cdot x^{4} \cdot y - 31.687374 \cdot x^{3} \cdot y^{2} - 11.876397 \cdot x^{2} \cdot y^{3} - 16.015567 \cdot x \cdot y^{4} - 2.549349 \cdot y^{5} + 25.993643 \cdot x^{4} + 18.066015 \cdot x^{3} \cdot y + 24.268005 \cdot x^{2} \cdot y^{2} + 11.100776 \cdot x \cdot y^{3} + 53.216611 \cdot y^{4} + 3.59729 \cdot x^{3} + 5.936538 \cdot x^{2} \cdot y + 3.004211 \cdot x \cdot y^{2} + 5.123685 \cdot y^{3} + 11.34886 \cdot x^{2} + 2.645745 \cdot x \cdot y + 13.059074 \cdot y^{2} + 0.797674 \cdot x + 1.733336 \cdot y + 13.661955$

\item $\Bb_{2,2}(x,y) := 1.722178 \cdot x^{6} + 5.310932 \cdot x^{5} \cdot y + 18.349052 \cdot x^{4} \cdot y^{2} - 31.02038 \cdot x^{3} \cdot y^{3} + 53.695332 \cdot x^{2} \cdot y^{4} - 8.244383 \cdot x \cdot y^{5} + 2.053936 \cdot y^{6} - 12.174742 \cdot x^{5} - 25.506336 \cdot x^{4} \cdot y - 24.443983 \cdot x^{3} \cdot y^{2} - 6.502687 \cdot x^{2} \cdot y^{3} - 11.222995 \cdot x \cdot y^{4} - 0.641161 \cdot y^{5} + 19.955712 \cdot x^{4} + 10.419893 \cdot x^{3} \cdot y + 20.713391 \cdot x^{2} \cdot y^{2} + 21.301025 \cdot x \cdot y^{3} + 42.607373 \cdot y^{4} + 2.265752 \cdot x^{3} + 5.151999 \cdot x^{2} \cdot y + 2.673623 \cdot x \cdot y^{2} + 4.233808 \cdot y^{3} + 9.995201 \cdot x^{2} + 2.756518 \cdot x \cdot y + 9.961636 \cdot y^{2} + 1.05112 \cdot x + 1.362802 \cdot y + 11.696231$\\

\item $\Bb_{2,3}(x,y) := 1.830446 \cdot x^{6} + 7.207784 \cdot x^{5} \cdot y + 8.603092 \cdot x^{4} \cdot y^{2} - 10.95115 \cdot x^{3} \cdot y^{3} + 44.535124 \cdot x^{2} \cdot y^{4} - 15.92115 \cdot x \cdot y^{5} - 3.079411 \cdot y^{6} - 10.209742 \cdot x^{5} - 23.847993 \cdot x^{4} \cdot y - 14.326739 \cdot x^{3} \cdot y^{2} - 3.364347 \cdot x^{2} \cdot y^{3} - 7.604397 \cdot x \cdot y^{4} + 0.614566 \cdot y^{5} + 12.955935 \cdot x^{4} + 8.069967 \cdot x^{3} \cdot y + 22.039468 \cdot x^{2} \cdot y^{2} + 27.849717 \cdot x \cdot y^{3} + 33.681351 \cdot y^{4} + 0.812285 \cdot x^{3} + 5.228008 \cdot x^{2} \cdot y + 2.414385 \cdot x \cdot y^{2} + 3.676612 \cdot y^{3} + 8.583767 \cdot x^{2} + 2.504491 \cdot x \cdot y + 7.086297 \cdot y^{2} + 1.186848 \cdot x + 0.995341 \cdot y + 9.732926$
\end{description}

\section{CBBC for Section 
\ref{subec:case_study_Kuramoto}}
\label{ap:coeffs_kuram}
The CBBC $\Bb_{\ell,i}(x)$ are defined picewise over and $i \in [0..2]$ as:

\begin{description}
    \item[] $\Bb_{0}(x,y,z) =-0.110879 \cdot x^{4} - 0.092034 \cdot x^{3} \cdot y - 0.082008 \cdot x^{3} \cdot z + 0.109202 \cdot x^{2} \cdot y^{2} - 1.7e-5 \cdot x^{2} \cdot y \cdot z + 0.061129 \cdot x^{2} \cdot z^{2} - 0.091925 \cdot x \cdot y^{3} - 4.0e-6 \cdot x \cdot y^{2} \cdot z + 3.4e-5 \cdot x \cdot y \cdot z^{2} - 0.082014 \cdot x \cdot z^{3} - 0.127569 \cdot y^{4} - 0.025422 \cdot y^{3} \cdot z - 0.022908 \cdot y^{2} \cdot z^{2} - 0.025878 \cdot y \cdot z^{3} - 0.097348 \cdot z^{4} + 1.774397 \cdot x^{3} - 0.000857 \cdot x^{2} \cdot y + 0.000538 \cdot x^{2} \cdot z - 0.000318 \cdot x \cdot y^{2} - 0.001115 \cdot x \cdot y \cdot z + 0.000179 \cdot x \cdot z^{2} + 1.827821 \cdot y^{3} - 0.000432 \cdot y^{2} \cdot z - 0.000194 \cdot y \cdot z^{2} + 1.632205 \cdot z^{3} + 0.089029 \cdot x^{2} + 0.706954 \cdot x \cdot y + 0.616631 \cdot x \cdot z + 0.301684 \cdot y^{2} + 0.207754 \cdot y \cdot z + 0.287124 \cdot z^{2} - 10.00198 \cdot x - 10.300025 \cdot y - 9.201995 \cdot z - 0.322743$ \\
    \item $\Bb_{1}(x,y,z) = -0.286681 \cdot x^{4} - 0.208624 \cdot x^{3} \cdot y - 0.205306 \cdot x^{3} \cdot z + 0.282097 \cdot x^{2} \cdot y^{2} - 0.000111 \cdot x^{2} \cdot y \cdot z + 0.174779 \cdot x^{2} \cdot z^{2} - 0.208234 \cdot x \cdot y^{3} + 5.8e-5 \cdot x \cdot y^{2} \cdot z + 3.2e-5 \cdot x \cdot y \cdot z^{2} - 0.205275 \cdot x \cdot z^{3} - 0.343453 \cdot y^{4} - 0.02198 \cdot y^{3} \cdot z - 0.107036 \cdot y^{2} \cdot z^{2} - 0.025665 \cdot y \cdot z^{3} - 0.193375 \cdot z^{4} + 4.496906 \cdot x^{3} - 0.002283 \cdot x^{2} \cdot y + 0.001331 \cdot x^{2} \cdot z - 0.001123 \cdot x \cdot y^{2} - 0.002579 \cdot x \cdot y \cdot z + 0.00018 \cdot x \cdot z^{2} + 4.692951 \cdot y^{3} - 0.001385 \cdot y^{2} \cdot z + 2.6e-5 \cdot y \cdot z^{2} + 3.083638 \cdot z^{3} + 0.199697 \cdot x^{2} + 1.606444 \cdot x \cdot y + 1.543309 \cdot x \cdot z + 0.921929 \cdot y^{2} + 0.212692 \cdot y \cdot z + 0.586415 \cdot z^{2} - 25.347415 \cdot x - 26.445779 \cdot y - 17.384738 \cdot z + 36.288487$ \\
    \item $\Bb_2(x,y,z) = -0.310531 \cdot x^{4} - 0.192396 \cdot x^{3} \cdot y - 0.192537 \cdot x^{3} \cdot z + 0.253625 \cdot x^{2} \cdot y^{2} - 0.000124 \cdot x^{2} \cdot y \cdot z + 0.140886 \cdot x^{2} \cdot z^{2} - 0.191943 \cdot x \cdot y^{3} + 8.2e-5 \cdot x \cdot y^{2} \cdot z + 2.5e-5 \cdot x \cdot y \cdot z^{2} - 0.192503 \cdot x \cdot z^{3} - 0.371679 \cdot y^{4} + 0.001279 \cdot y^{3} \cdot z - 0.158635 \cdot y^{2} \cdot z^{2} - 0.002735 \cdot y \cdot z^{3} - 0.20944 \cdot z^{4} + 4.788563 \cdot x^{3} - 0.00241 \cdot x^{2} \cdot y + 0.001301 \cdot x^{2} \cdot z - 0.001302 \cdot x \cdot y^{2} - 0.002351 \cdot x \cdot y \cdot z + 8.0e-5 \cdot x \cdot z^{2} + 5.016419 \cdot y^{3} - 0.001557 \cdot y^{2} \cdot z + 6.3e-5 \cdot y \cdot z^{2} + 3.264085 \cdot z^{3} + 0.405335 \cdot x^{2} + 1.486829 \cdot x \cdot y + 1.447237 \cdot x \cdot z + 1.175779 \cdot y^{2} + 0.04123 \cdot y \cdot z + 0.806732 \cdot z^{2} - 26.991353 \cdot x - 28.26892 \cdot y - 18.40208 \cdot z + 51.574377$
\end{description}

\end{document}